\documentclass[11pt,reqno]{amsart}
\usepackage[utf8]{inputenc}
 
\usepackage{amsmath,amsthm,amssymb}
\usepackage{mathabx}

\usepackage{graphicx}
\usepackage{caption,subcaption}
\usepackage{tikz}
\usepackage{geometry}
\usepackage{color}
\usepackage{cleveref}
\usepackage[sort&compress,numbers]{natbib}
\usepackage{todonotes}
\usepackage{comment}

\usepackage{pgfplots}
\usetikzlibrary{intersections,decorations.pathreplacing,decorations.markings,calc,angles,quotes,arrows.meta,pgfplots.fillbetween,patterns}

\newtheorem{theorem}{Theorem}

\newtheorem{prop}[theorem]{Proposition}
\newtheorem{lemma}[theorem]{Lemma}
\newtheorem{cor}[theorem]{Corollary}
\newtheorem{remark}[theorem]{Remark}

\numberwithin{equation}{section}
\numberwithin{theorem}{section}

\newcommand{\coup}{{g}}
\newcommand{\iu}{\mathrm{i}\mkern1mu} 
\renewcommand{\d}{\mathrm{d}}
\newcommand{\p}{\partial}

\newcommand{\bm}{{\bf m}}
\newcommand{\bq}{{\bf q}}

\newcommand{\bv}{{\bf v}}

\newcommand{\bk}{{\bf k}}

\newcommand{\bx}{{\bf x}}
\newcommand{\by}{{\bf y}}
\newcommand{\bz}{{\bf z}}

\newcommand{\tr}{\operatorname{tr}}

\newcommand{\R}{\mathbb{R}}

\newcommand{\Z}{\mathbb{Z}}
\newcommand{\C}{\mathbb{C}}
\newcommand{\T}{\mathbb{T}}

\newcommand{\A}{\mathcal{A}}
\newcommand{\B}{\mathcal{B}}
\newcommand{\E}{\mathcal{E}}

\newcommand{\G}{\mathcal{G}}

\renewcommand{\L}{\mathcal{L}}

\newcommand{\ds}{\displaystyle}

\DeclareMathOperator\supp{supp}

\renewcommand\labelenumi{(\roman{enumi})}
\renewcommand\theenumi\labelenumi

\title[Nonlocal PDEs and Quantum Optics: band structure]{Nonlocal partial
differential equations and quantum optics: band structure of periodic atomic sytems}

\author[Hiltunen]{Erik Orvehed Hiltunen}
\address{Department of Mathematics, Yale University, New Haven, CT, USA }
\email{erik.hiltunen@yale.edu}

\author[Kraisler]{Joseph Kraisler}
\address{Department of Mathematics and Statistics, Amherst College, Amherst, MA, USA}
\email{jkraisler@amherst.edu}

\author[Schotland]{John C. Schotland}
\address{Department of Mathematics and Department of Physics, Yale University, New Haven, CT, USA}
\email{john.schotland@yale.edu}

\author[Weinstein]{Michael I. Weinstein}
\address{Department of Applied Physics and Applied Mathematics and Department of Mathematics, Columbia University, New York, NY, USA}
\email{miw2103@columbia.edu}
\date{\today}

\begin{document}

\maketitle

\begin{abstract}
We continue our study of the quantum optics of a single photon interacting with a system of two level atoms. In this work we investigate the case of a periodic arrangement of atoms. We provide a general structure theorem characterizing the band functions of this problem, which comprise the spectrum of the associated Hamiltonian. Additionally, we study atomic densities arising as periodically arranged scaled inclusions. For this family of examples, we obtain explicit asymptotic formulas for the band functions. 
\end{abstract}
\section{Introduction}
This paper is the second in a series that is concerned with a class of mathematical problems that arise in quantum optics. Ref.~\cite{one-photon_bound} is the first paper in the series, which we will refer to as part I.

We consider a system of two-level atoms coupled to a quantized field. Following~\cite{kraisler_2022}, the dynamics of a single-excitation state is governed by the equations
\begin{align}
\label{eq:b4}
\iu\partial_t\psi & = c(-\Delta)^{1/2}\psi + \coup\rho(\bx) a , \\
\iu \rho(\bx)\partial_t a & = \coup\rho(\bx)\psi + \Omega\rho(\bx) a ,
\label{eq:b5}
\end{align}
for $\bx \in \R^d$, $d\ge 1$. Here $\psi(\bx,t)$ and $a(\bx,t)$ denote the probability amplitudes for creating a photon and exciting an atom, respectively at time $t$. In addition, $c$ is the speed of light, $g$ is the atom-field coupling strength, $\Omega$  is the resonance frequency of the atoms, and $\rho$ is the number density of the atoms. The fractional Laplacian $(-\Delta)^{1/2}$  is a nonlocal operator and admits the following representation,
\begin{align}
    (-\Delta)^{1/2}f = \int_{\R^d} e^{\iu \bx\cdot\bk}\vert\bk\vert \hat f(\bk)\frac{\d \bk}{(2\pi)^d} ,
\end{align}
where the Fourier transform $\hat f(\bk)$ of a function $f(\bx)\in L^2(\R^d)\cap L^1(\R^d)$ is defined by
\begin{align}
    \hat f(\bk) = \int_{\R^d} e^{-\iu\bk\cdot\bx} f(\bx) \d \bx.
\end{align}
Eqs.~\eqref{eq:b4} and \eqref{eq:b5} have been studied in a variety of physical contexts, depending on the choice of $\rho$~\cite{kraisler_2022, hoskins_2021,hoskins_2023,kraisler2022kinetic}. These include spontaneous emission by a single atom, systems of many atoms and random media. This theme was continued in Part I, which focused on the characterization of bound states and resonances of \eqref{eq:b4} and \eqref{eq:b5}. 
 
For the remainder of this paper, we consider the time-harmonic solutions to the system \eqref{eq:b4}--\eqref{eq:b5} of the form
\begin{align}\label{eq:t-harmonic}
    \begin{pmatrix}
    a(\bx,t) \\ \psi(\bx,t) 
    \end{pmatrix} = e^{-\iu\omega t}\begin{pmatrix}
    a(\bx) \\\psi(\bx)
    \end{pmatrix}.
\end{align}
We thus obtain the spectral problem
\begin{align} 
	c(-\Delta)^{1/2}\psi + g\rho(\bx)a &= \omega \psi, \label{eq:c4} \\
	g\rho(\bx) \psi+ \Omega \rho(\bx) a &= \omega \rho(\bx) a . \label{eq:c5}
\end{align}
Since the second equation is algebraic, we can eliminate $a$ and obtain the nonlinear eigenvalue problem
\begin{equation}\label{eq:EVproblem}
	\left((-\Delta)^{1/2} - \frac{\omega}{c}\right) \psi = - \frac{g^2}{c(\omega-\Omega)}\rho(\bx)\psi.
\end{equation}
valid for $\omega\ne\Omega$. We have thus reduced the system \eqref{eq:c4}--\eqref{eq:c5} to a single equation \eqref{eq:EVproblem}, where $a$ can be recovered by the formula
\begin{align}\label{eq:a_psi}
    a(\bx) = \frac{g}{\omega-\Omega}\psi(\bx),\ \ \bx\in\supp(\rho).
\end{align}

This paper is concerned with the analysis of \eqref{eq:c4} and \eqref{eq:c5} when $\rho(\bx)$ is a periodic function with respect to a crystal lattice. The motivation is to investigate quantum optical effects in photonic band gap materials. Such materials have long been of interest due to their ability to produce localized quantum states \cite{yablonovitch1993photonic,john1991quantum}. Moreover, recent progress in the field of topological photonics provides an avenue for achieving photon localization or propagation which is robust to disorder \cite{barik2018topological,perczel2017topological}. Of central importance in this endeavor are analytical and computational methods for studying the band structure of periodic Hamiltonians.

The main result of this paper, \Cref{theorem_1}, shows that there are two families $\omega_j^\pm(\bk)$ of band functions. One family with $\omega_j^+(\bk) > \Omega$ is increasing to $+\infty$, while the other family $\omega_j^-(\bk) < \Omega$ is increasing and accumulates at $\Omega$. To illustrate our approach, we  begin by reviewing the case of constant density originally studied in \cite{kraisler_2022}, before presenting the  proof of \Cref{theorem_1} in \Cref{sec:proof}. In \Cref{sec:bounds}, we discuss general bounds on $\omega_j^\pm(\bk)$, which show the existence of a band gap above $\Omega$ for certain $\rho$. With this general theory in hand, we proceed with a detailed study of the case of high contrast atomic inclusions $\rho(\bx)$ in \Cref{sec:highcontrast} and \Cref{sec:nonlin}. The case of a single high-contrast inclusion was treated in part I, where explicit asymptotic expansions of the resonances were derived. Specifically, in part I,  a sequence of resonances were described as perturbations of the eigenvalues of a limiting operator which is self-adjoint, compact, and positive-definite. In \Cref{sec:highcontrast}, we show that the family of bands $\omega_j^-(\bk)$ can be described as perturbations of the same sequence of eigenvalues. Finally, in \Cref{sec:nonlin}, we develop a formal method for computing the family $\omega_j^+(\bk)$. This family of band functions arises as perturbations of the spectrum of the free problem.

\section{Floquet-Bloch Theory}
In this section we outline the Floquet-Bloch theory used throughout the paper. Let $\Lambda\subset \R^d$ be a lattice spanned by a collection of linearly independent vectors 
\begin{align}
    \bv_1,\cdots,\bv_d\in \R^d .
\end{align}
Equivalently, we have
\begin{align}
    \Lambda = \Z\bv_1\oplus\cdots\oplus \Z\bv_d.
\end{align}
We denote by $\T^d = \R^d/\Lambda$ the $d$ dimensional torus with respect to this lattice, while $Y$ denotes a fundamental cell for $\Lambda$ defined as
\begin{align}
    Y = \left\{ \sum_{i=1}^d \theta_i\bv_i \ | \ 0\leq \theta_i\leq 1\right\} .
\end{align}
The space of $\bk$-pseudo-periodic $L^2$ functions is given by
\begin{align}
     L^2_{\bk,\Lambda} = \{f\in L^2_{\text{loc}}\ | \ f(\bx+\bv) = e^{\iu \bk\cdot\bv}f(\bx), \bv\in\Lambda  \}.
\end{align}
Often we will drop the explicit dependence on the lattice and simply refer to this space as $L^2_{\bk}$. In the case that $\bk=0$, the pseudo-periodic functions are periodic. Correspondingly, we denote the space of periodic $L^2$ functions as $L^2_{\text{per},\Lambda}=L^2(\T^d)$. For $f,g\in L^2_{\bk,\Lambda}$ we define the inner product 
\begin{align}
    \langle f,g\rangle = \int_{Y} f^*(\bx)g(\bx) \d\bx .
\end{align}
The dual lattice $\Lambda^*$ is defined
\begin{align}
    \Lambda^* = \Z\bk_1\oplus\cdots\oplus\Z\bk_d ,
\end{align}
where the vectors $\bk_1,\cdots,\bk_d\in \R^d$ are linearly independent and chosen such that
\begin{align}
    \bk_i\cdot\bv_j = 2\pi\delta_{ij}.
\end{align}
The first Brillouin zone, denoted $\mathcal{B}$, is a choice of representatives of the dual torus $\R^d /\Lambda^*$ defined as the collection of points $\bk\in\R^d$ which are closer to the origin than any other lattice point in $\Lambda^*$.

If $H$ is a self-adjoint operator on $L^2(\R^d)$ then for each $\bk\in\R^d$ we can consider the Floquet-Bloch eigenvalue problem
\begin{align}
    H\Psi(\bx,\bk) &= \mu(\bk)\Psi(\bx,\bk) , \\
    \Psi(\bx+\bv,\bk) &= e^{\iu \bv\cdot\bk}\Psi(\bx,\bk), \quad \bv\in\Lambda .
\end{align}
As this problem is invariant under the translation $\bk \to \bk+\bk'$ for any $\bk'\in \Lambda^*$ it suffices to study the problem for $\bk\in\mathcal{B}$. Alternatively, one may make the substitution 
\begin{align}
    \Phi(\bx,\bk) = e^{-\iu \bk\cdot\bx}\Psi(\bx,\bk),
\end{align}
and obtain a problem with periodic boundary conditions
\begin{align}\label{eq:FBproblem}
    H(\bk)\Phi(\bx,\bk) &= \mu(\bk)\Phi(\bx,\bk) ,\\
    \Phi(\bx+\bv,\bk) & = \Psi(\bx,\bk) \quad \bv\in \Lambda ,
\end{align}
where the $\bk$-dependent operator $H(\bk)$ is defined as
\begin{align}
    H(\bk) = e^{-\iu \bk\cdot\bx}H e^{\iu \bk\cdot\bx} .
\end{align}
Moreover the spectrum of $H$ in $L^2(\R^d)$ is the union of the intervals swept out by the band functions $\mu_j(\bk)$ as $\bk$ varies over $\mathcal{B}$:
\begin{align}
    \sigma(H) = \bigcup_{\bk\in\mathcal{B}} \sigma\bigl(H(\bk)\bigr)  .
\end{align}

\section{General spectral theory}\label{sec:general}
In this section we develop the general spectral theory of the operator introduced in the eigenvalue problem \eqref{eq:EVproblem}. Let us define the Hamiltonian $H$
\begin{align}
    H = H_0 + g V,
\end{align}
acting on a dense subspace of $L^2(\R^d;\C^2)$, where
\begin{align}
    H_0 &=\begin{pmatrix} (-\Delta)^{1/2} & 0 \\ 0 & \Omega
    \end{pmatrix}, \\
    V &= \begin{pmatrix} 0 & \sqrt{\rho(\bx)} \\ \sqrt{\rho(\bx)} & 0
    \end{pmatrix}. \label{eq:V}
\end{align}
Here the atomic density $\rho(\bx)\in L^{\infty}(\R^d)$ is nonnegative ($\rho(\bx)\geq 0)$. Moreover, we suppose that $\rho(\bx)$ is periodic with respect to some lattice $\Lambda\in\R^d$. We note that $H_0$ and $V$ are self-adjoint and commute with translations by any vector $\bv\in \Lambda$. By the Floquet-Bloch theory in Section 2, we consider the eigenvalue problem on the space of pairs of periodic functions
\begin{align}
    H(\bk)\Phi(\bx,\bk) = \omega(\bk)\Phi(\bx,\bk), \label{eq:EVprob1}
\end{align}
for $\Phi(\bx,\bk)\in L_{\bk}^2(\R^d;\C^2)$. Explicitly, $H(\bk)=H_0(\bk)+g V$ where $V$ is given by \eqref{eq:V} and $H_0(\bk)$ is given by
\begin{align}
    H_0(\bk)=\begin{pmatrix} (-(\nabla+\iu\bk)^2)^{1/2} & 0 \\ 0 & \Omega
    \end{pmatrix}.
\end{align}

We now state the main theorem of this section, whose proof will be given in \Cref{sec:proof}. It provides a structure of the spectrum $\sigma\bigl(H(\bk)\bigr) $ of the Bloch Hamiltonians $H(\bk)$.
\begin{theorem}\label{theorem_1}
The spectrum of $H(\bk)$ acting on $L^2(\T^d;\C^2)$ is discrete in $\R\setminus\{\Omega\}$ with
    \begin{enumerate}
        \item $\sigma\bigl(H(\bk)\bigr) \cap(\Omega,\infty)$ comprised of a sequence of eigenvalues $0\leq\omega^+_1(\bk)\leq \omega^+_2(\bk) \leq \cdots$ increasing with $\lim_{j\to\infty}\omega^+_j(\bk)=\infty$.
        \item $\sigma\bigl(H(\bk)\bigr) \cap(-\infty,\Omega)$ comprised of a sequence of eigenvalues $\omega^-_1(\bk)\leq\omega^-_2(\bk)\leq\cdots$ increasing and accumulating at $\Omega$, i.e. $\lim_{j\to\infty}\omega^-_j(\bk)=\Omega$.
        \item The corresponding collection of eigenvectors form a basis for $L^2(\T^d;\C^2)$.
    \end{enumerate}
\end{theorem}
To prove this theorem, we introduce a $2$-parameter family of operators $H(\bk,\omega)$ on $L^2(\T^d)$ by
\begin{align}
    H(\bk,\omega)\phi = H_0(\bk)\phi &+ \frac{g^2\rho(\bx)}{\omega-\Omega}\phi\, , \label{eq:2parameterHamiltonian}
\end{align}
where $H_0(\bk) =  (-(\nabla+\iu\bk)^2)^{1/2}$, and denote the eigenvalues of $H(\bk,\omega)$ by $\lambda_m(\bk,\omega)$. The purpose is to reduce the nonlinear eigenvalue problem considered in \Cref{eq:EVprob1} to a family of linear eigenvalue problems, each depending on a parameter $\omega$. Once this family of eigenvalue problems are well understood, we can use them to solve the original problem.

The following lemma describes the precise relationship between the eigenvalues of $H(\bk)$ and those of $H(\bk,\omega)$.
\begin{lemma}\label{lemma:eigenvalue}
$\omega\in\R\setminus \{\Omega\}$ is an $L^2(\T^d;\C^2)$ eigenvalue of $H(\bk)$ if and only if $\omega$ is an $L^2(\T^d)$ eigenvalue of $H(\bk,\omega)$.
\end{lemma}
\begin{proof}
Suppose that we have a solution $\Phi=(\phi_1,\phi_2)\in L^2(\T^d;\C^2)$ to the equation
\begin{align}
    H(\bk)\Phi=\omega\Phi.
\end{align}
By solving for $\phi_2$ in terms of $\phi_1$, specifically
\begin{align}
    \phi_2 &= \frac{g\sqrt{\rho(\bx)}}{\omega-\Omega}\phi_1,
\end{align}
and substituting into the first equation, we have
\begin{align}
      (-(\nabla+\iu\bk)^2)^{1/2}\phi_1 &+ \frac{g^2\rho(\bx)^2}{\omega-\Omega}\phi_1 = \omega\phi_1.
\end{align}
This is the statement that
\begin{align}
    H(\bk,\omega)\phi_1=\omega\phi_1.
\end{align}
Conversely, if we have a solution to
\begin{align}
    H(\bk,\omega)\phi = \omega\phi,
\end{align}
with $\phi\in L^2(\T^d)$ and $\omega \neq\Omega$, then we can define
\begin{align}
    \psi = \frac{g\sqrt{\rho(\bx)}}{\omega-\Omega}\phi.
\end{align}
Then $\psi\in L^2(\T^d)$ as $\rho(\bx)$ is almost surely bounded. Moreover, the pair $(\phi,\psi)$ is readily seen to be a solution to
\begin{align}
    H(\bk)(\phi,\psi) = \omega(\phi,\psi),
\end{align}
which proves the claim.
\end{proof} 
\subsection{Case study: constant density} \label{sec:const}
Before we prove the theorem, we will study a simple, yet instructive example: the constant density case $\rho(\bx)\equiv \rho_0$. For clarity we will work in the one dimensional case of periodicity $1$. Then the eigenvalue problem for $H(k,\omega)$, defined in \Cref{eq:2parameterHamiltonian}, becomes
\begin{align}
    \left(-\left(\frac{\d}{\d x}+\iu k\right)^2\right)^{1/2}\phi + \frac{g^2\rho_0}{\omega-\Omega}\phi = \lambda\phi.
\end{align}
By introducing the Fourier transform
\begin{align}
    \hat{\phi}(m) = \int_0^1 e^{2\pi \iu m x}\phi(x)\d x,
\end{align}
we see that the Fourier coefficients must satisfy
\begin{align}
    \vert 2\pi m + k\vert\hat\phi + \frac{g^2\rho_0}{\omega-\Omega}\hat\phi = \lambda\hat\phi.
\end{align}
Thus we have that the eigenvalues $\lambda_m(k,\omega)$ as a function of $k$ and $\omega$ are given by
\begin{align}
    \lambda_m(k,\omega)= \vert 2\pi m + k\vert + \frac{g^2\rho_0}{\omega-\Omega}.
\end{align}
We note that the curves $\lambda_m(k,\omega)$ 
\begin{enumerate}
    \item are increasing in $m$;
    \item are continuous for $\omega\neq \Omega$;
    \item have a vertical asymptote at $\omega=\Omega$;
    \item asymptotically approach the eigenvalues of the free problem $\rho_0=0$ as $\omega\to \pm\infty$.
\end{enumerate}
We will see in the proof of \Cref{theorem_1} that these features are generic. Additionally, the functions $\lambda_m(k,\omega)$ are concave up for $\omega>\Omega$ and concave down for $\omega<\Omega$. We recover the eigenvalues of $H(k)$ by considering the intersection of $\lambda_m(k,\omega)$ with the diagonal, namely 
\begin{align}
    \omega =  \vert 2\pi m + k\vert + \frac{g^2\rho_0}{\omega-\Omega}.
\end{align}
This is quadratic in $\omega$ and the solutions are given by
\begin{align}\label{eq:band_const}
    \omega_{m,\pm}(k) = \frac{1}{2}\left[\Omega+ \vert 2\pi m + k\vert\pm\sqrt{\bigl(\Omega- \vert 2\pi m + k\vert\bigr)^2+4g^2\rho_0}\right],
\end{align}
from which we see
\begin{align}
    \frac{\Omega + \sqrt{\Omega^2 + 4g^2\rho_0 }}{2} & \leq \omega^+_{m}(k), \\
   \frac{\Omega - \sqrt{\Omega^2 + 4g^2\rho_0 }}{2} &\leq  \omega^-_{m}(k) < \Omega. \label{eq:bound_const}
\end{align}
\begin{remark}
The labeling of the solutions $\omega_{m,\pm}$ does not directly correspond to the labeling of the eigenvalues $\omega_{m}^{\pm}$ in \Cref{theorem_1}. The precise relationship is given below:
\begin{align}
    \omega_{m}^{+}(k) =\begin{cases} \omega_{m-1,+}(k) & 0 \leq k < \pi \\ \omega_{-m,+}(k) & \pi \leq k <2\pi
    \end{cases}, \quad \text{or} \quad
       \omega_{m}^{+}(k) =\begin{cases} \omega_{-m,+}(k) & 0 \leq k < \pi \\ \omega_{m-1,+}(k) & \pi \leq k <2\pi
    \end{cases}, 
\end{align}
for $m$ odd or $m$ even, respectively. Similarly, we may also identify
\begin{align}
    \omega_{m}^{-}(k)=\begin{cases} \omega_{m-1,-}(k) & 0 \leq k < \pi \\ \omega_{-m,-}(k) & \pi \leq k <2\pi
    \end{cases},  \quad \text{or} \quad
   \omega_{m}^{-}(k) =\begin{cases} \omega_{-m,-}(k) & 0 \leq k < \pi \\ \omega_{m-1,-}(k) & \pi \leq k <2\pi
    \end{cases},
\end{align}
for $m$ odd or $m$ even, respectively.
\end{remark}
As the constant coefficient case may also be solved using the Fourier transform, the band functions $\omega_{m}^{\pm}(k)$ obtained this way must agree with the dispersion relation obtained from the constant density problem
\begin{align*}
	c(-\Delta)^{1/2}\psi + g\rho_0 a &= \omega \psi, \\
	g\psi+ \Omega  a &= \omega a.
\end{align*}
As this is a $2\times 2$ constant density system, we obtain two curves
\begin{align}\label{eq:disp}
  \omega_{\pm}(k) = \frac{1}{2}\left[c\vert k\vert +\Omega \pm\sqrt{(c\vert k\vert-\Omega)^2+4g^2\rho_0^2}\right],  
\end{align}
which we have plotted for several values of the coupling constant $g$ in \Cref{fig:const}. One obtains the bands $\omega_{m}^{\pm}(k)$ by defining $k$ to be on the torus $\T$ and ''folding" \eqref{eq:disp} into the Brillouin zone, which agrees with \eqref{eq:band_const}. This folding process is illustrated in \Cref{fig:fold} which shows the bands $\omega_{m}^{\pm}(k)$ as in \eqref{eq:band_const}, which agree with the "unfolded" bands shown in \Cref{fig:const}.
\begin{figure}
	\begin{subfigure}[t]{0.33\linewidth}
\begin{tikzpicture}[scale = 0.6]
	\begin{axis}[
		axis lines=middle,
		clip=false,
		xticklabels=\empty,
		yticklabels=\empty,
		xlabel style = {at={(1.01,0.19)}},
		ylabel style = {at={(0.5,0.97)},right},
		ylabel=$\omega(\bk)$,
		xlabel=$\bk$,		
		xmin=-5,
		xmax=5,
		ymin=-2.5,
		ymax=6.5,
		]
		\addplot+[red,solid,mark=none,samples=200,unbounded coords=jump] {(1+abs(x))/2 + sqrt((1-abs(x))^2/4 + 0.01)};
		\addplot+[red,solid,mark=none,samples=200,unbounded coords=jump] {(1+abs(x))/2 - sqrt((1-abs(x))^2/4 + 0.01)};
	\end{axis}
\end{tikzpicture}
\caption{$g=0.1$}
\end{subfigure}\hfill
\begin{subfigure}[t]{0.33\linewidth}
\begin{tikzpicture}[scale=0.6]
		\begin{axis}[
		axis lines=middle,
		clip=false,
		xticklabels=\empty,
		yticklabels=\empty,
		xlabel style = {at={(1.01,0.19)}},
		ylabel style = {at={(0.5,0.97)},right},
		ylabel=$\omega(\bk)$,
		xlabel=$\bk$,		
		xmin=-5,
		xmax=5,
		ymin=-2.5,
		ymax=6.5,
		]
		\addplot+[blue,solid,mark=none,samples=200,unbounded coords=jump] {(1+abs(x))/2 + sqrt((1-abs(x))^2/4 + 1)};
		\addplot+[blue,solid,mark=none] {(1+abs(x))/2 - sqrt((1-abs(x))^2/4 + 1)};
	\end{axis}
\end{tikzpicture}
\caption{$g=1$}\label{fig:B}
\end{subfigure}\hfill
\begin{subfigure}[t]{0.33\linewidth}
\begin{tikzpicture}[scale=0.6]
	\begin{axis}[
		axis lines=middle,
		clip=false,
		xticklabels=\empty,
		yticklabels=\empty,
		xlabel style = {at={(1.01,0.19)}},
		ylabel style = {at={(0.5,0.97)},right},
		ylabel=$\omega(\bk)$,
		xlabel=$\bk$,		
		xmin=-5,
		xmax=5,
		ymin=-2.5,
		ymax=6.5,
		]
		\addplot+[green,solid,mark=none,samples=200,unbounded coords=jump] {(1+abs(x))/2 + sqrt((1-abs(x))^2/4 + 10)};
		\addplot+[green,solid,mark=none,samples=200,unbounded coords=jump] {(1+abs(x))/2 - sqrt((1-abs(x))^2/4 + 10)};
	\end{axis}
\end{tikzpicture}
\caption{$g=10$}
\end{subfigure}
\caption{Dispersion relations of constant-density materials with $\Omega = 1$, for three different values of $g$. In all cases, the lower band has a vertical asymptote at $\Omega$. When the coupling $g$ is small, the two bands tend to the light cone and a flat band at $\Omega$.}\label{fig:const}
\end{figure}
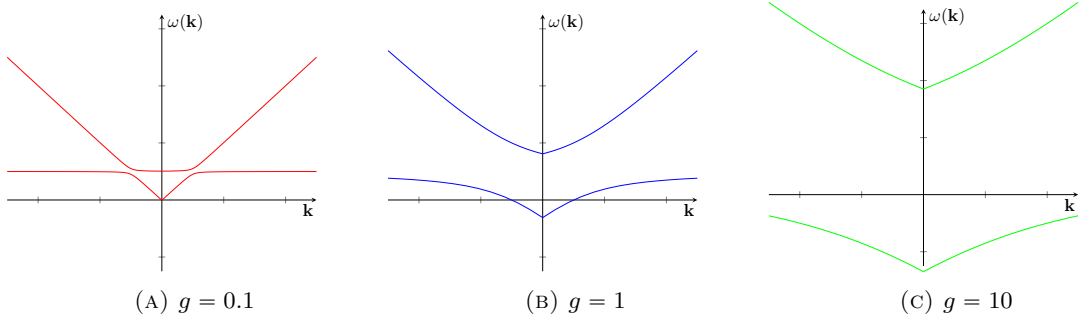


\begin{figure}
	\begin{subfigure}[t]{0.5\linewidth}
	\includegraphics[width=0.8\linewidth]{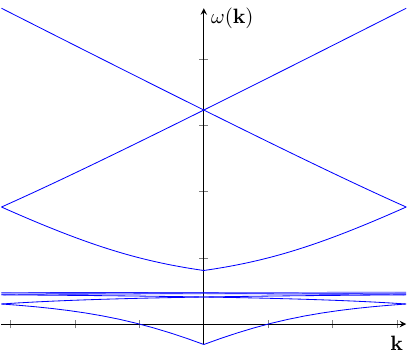}
	\caption{$\omega_j^-$ and $\omega_j^+$}
	\end{subfigure}\hfill
\begin{subfigure}[t]{0.5\linewidth}
	\includegraphics[width=0.8\linewidth]{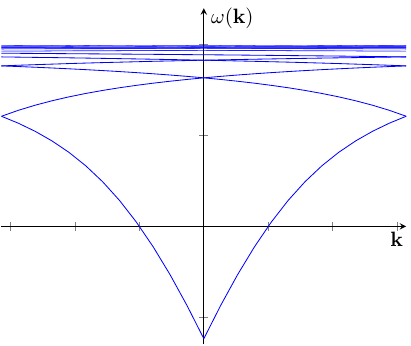}
	\caption{Close-up showing $\omega_j^-$}
\end{subfigure}
	\caption{Band functions of a constant-density material with $\Omega = 1$ and $g=1$. We emphasize that these band functions emerge solely from folding of the dispersion relationship seen in \Cref{fig:B}. The vertical asymptote seen in \Cref{fig:const} is now folded into a family of bands $\omega^-(\bk)$ which accumulate at $\Omega$. Above $\Omega$ there is a band gap, followed by a second family of bands $\omega^+(\bk)$.} \label{fig:fold}
\end{figure}

\subsection{Proof of \Cref{theorem_1}} \label{sec:proof} In this section we formalize the ideas of \Cref{sec:const} which, after proving two preliminary results, will allow us to prove \Cref{theorem_1}. 

From \Cref{lemma:eigenvalue} we reduce the study of the eigenvalues of the one-parameter family of Hamiltonians $H(\bk)$ acting on $L^2(\T^d;\C^2)$ to the study of the two-parameter family of Hamiltonians $H(\bk,\omega)$ acting on $L^2(\T^d)$.

\begin{prop}\label{prop:properties}
Let $\bk\in\mathcal{B}$ be fixed. Then for each $\omega\neq\Omega$ the spectrum of $H(\bk,\omega)$ is discrete, bounded below, and increases to infinity:
\begin{align}
    \lambda_1(\bk,\omega)\leq \lambda_2(\bk,\omega)&\leq \cdots , \\
    \lim_{j\to\infty}\lambda_j(\bk,\omega) &= \infty.
\end{align}
\item For each $j$, the curve $\omega\mapsto\lambda_j(\bk,\omega)$ with $\bk$ fixed has the following properties:
\begin{enumerate}
    \item $\lambda_j(\bk,\omega)$ is continuous in $\omega\in \R\setminus\{\Omega\}$;
    \item For $\omega>\Omega$, $\lambda_j(\bk,\omega) \geq \lambda_j^0(\bk)$, while for $\omega < \Omega$  $\lambda_j(\bk,\omega) \leq \lambda_j^0(\bk)$, where $\lambda_j^0(\bk)$ is the $j$th eigenvalue of $H_0(\bk)$;
    \item $\lim_{\omega\to\pm\infty}\lambda_j(\bk,\omega)= \lambda^0_j(\bk)$;
    \item $\lim_{\omega\to\Omega^-}\lambda_j(\bk,\omega)=-\infty$.
\end{enumerate}
If $\rho(\bx)$ is bounded below by a positive constant, we have
\begin{enumerate}
    \setcounter{enumi}{4}
    \item $\lim_{\omega\to\Omega^+}\lambda_j(\bk,\omega)=\infty$.
\end{enumerate}
\end{prop}
\begin{proof}
Fix $\bk\in \mathcal{B}$. For $\omega \neq \Omega$ we define the functional
\begin{align}
    L(\bk,\omega)[\phi] &= \int_{\T^d}\vert (-(\nabla + \iu\bk)^2)^{1/4} \phi\vert^2 + \frac{g^2\rho(\bx)}{\omega-\Omega}\vert \phi\vert^2 \d\bx, 
\end{align}
with domain $D(L(\bk,\omega)) = H^{1/2}(\T^d)$. 
By standard max-min arguments applied to functions $\phi\in H^{1/2}(\T^d)$ with fixed $L^2(\T^d)$ norm
\begin{align}
    \int_{\T^d}\vert\phi(\bx)\vert^2 \d\bx = 1,
\end{align}
we find that there is a sequence $\{\lambda_{j}(\bk,\omega)\}_{j\geq 1}$ of eigenvalues increasing to $+\infty$. This proves the first part of the theorem. Moreover, each curve $\lambda_j(\omega)$ is Lipschitz continuous in $\omega$ which proves (i). Now define the functional $L_{0}(\bk)$, also with domain $H^{1/2}(\T^d)$, by
\begin{align}
    L_{0}(\bk)[\phi] = \int_{\T^d}\left| (-(\nabla + \iu\bk)^2)^{1/4} \phi\right|^2 \d\bx .
\end{align}
Then the $H^{1/2}(\bk)$ eigenvalues $\lambda_j^0(\bk)$ of $H_0(\bk)$ are obtained through the min-max principle applied to $L_{0}(\bk)$. We have the following inequalities
\begin{align}
    \begin{cases} L(\bk,\omega)[\phi] > L_0(\bk)[\phi], & \text{for} \quad \omega > \Omega,\\
    L(\bk,\omega)[\phi] < L_0(\bk)[\phi], & \text{for} \quad \omega < \Omega.
    \end{cases}
\end{align}
Thus, (ii) follows from a comparison principle arising from the min-max principle. For (iii) we note
\begin{align}
    \left\vert  L(\bk,\omega)[\phi] - L_0(\bk)[\phi] \right\vert = \frac{g^2}{\vert \omega-\Omega\vert}\int_{\T^d}\rho(\bx)\vert\phi(\bx)\vert^2\d\bx \to 0,
\end{align}
as $\omega\to\pm\infty$ for any $\phi\in H^{1/2}(\T^d)$ such that $\|\phi\|_{L^2} = 1$. For (iv) we exhibit an infinite orthogonal family of functions $\{\phi_n\}\subset H^{1/2}(\T^d)$ for which
\begin{align}
    \lim_{\omega\to\Omega^{-}}L(\bk,\omega)[\phi_n] = -\infty .
\end{align}
Fix $\bk'\in \Lambda^*$ and consider
\begin{align}
    \phi_n(\bx) = \frac{1}{\text{vol}(\T^d)^{1/2}}e^{\iu  n\bk'\cdot\bx} .
\end{align}
Then by direct calculation
\begin{align}
    L(\bk,\omega)[\phi_n] &= \vert n\bk'+\bk\vert + \frac{g^2}{\omega-\Omega}\frac{1}{\text{vol}(\T^d)}\int_{\T^d}\rho(\bx)\d\bx \\
    &= \vert n\bk'+\bk\vert -\frac{g^2 \tilde{\rho}}{\vert\omega-\Omega\vert} ,
\end{align}
where
\begin{align}
    \tilde{\rho} = \frac{1}{\text{vol}(\T^d)}\int_{\T^d}\rho(\bx)\d\bx > 0,
\end{align}
is the average value of $\rho$.

For (v), suppose there exists $m>0$ such that $m \leq \rho(\bx)$. Then for $\omega>\Omega$ we have
\begin{align}
    \lambda_j(\bk,\omega) &= \int_{\T^d}  \vert (-(\nabla + \iu\bk)^2)^{1/4}\phi_j\vert^2+\frac{g^2\rho(\bx)}{\omega-\Omega}\vert \phi_j\vert^2 \d\bx\\
    &\geq \int_{\T^d} \frac{g^2\rho(\bx)}{\omega-\Omega}\vert \phi_j\vert^2 \d\bx\\
    &\geq \frac{gm}{\omega-\Omega}\int_{\T^d} \vert \phi_j\vert^2 \d\bx = \frac{gm}{\omega-\Omega}.
\end{align}
\end{proof}
With \Cref{lemma:eigenvalue} and \Cref{prop:properties} at hand, we can now prove the main result, \Cref{theorem_1}.
\begin{proof}[Proof of \Cref{theorem_1}]
Consider $\lambda_j(\bk,\omega)$. By the proposition, it has a vertical asymptote as $\omega\to\Omega^{-}$ and asymptotically approaches $\lambda^0_j(\bk)$ as $\omega\to -\infty$. Thus there must be some intersection with the diagonal for $\omega<\Omega$. Call this $\omega^{-}_j(\bk)$. Additionally, if $\lambda^0_j(\bk) > \Omega$ (which is true for all but finitely many $\lambda^0_j(\bk)$), then as $\lambda_j(\bk,\omega) > \lambda^0(\bk)$ and asymptotically approaches $\lambda^0_j(\bk)$ as $\omega\to +\infty$, there must be some intersection with the diagonal for $\omega > \Omega$ as well. Call this $\omega^{+}_j(\bk)$. Then we have that $\omega_j^{+}$ (resp. $\omega_j^{-}$) are eigenvalues of $H(\bk,\omega^{+}_j)$ (resp. $H(\bk,\omega^{-}_j)$) hence the two sequences $\{\omega^{\pm}_n\}_{n\geq 1}$ are eigenvalues of $H(\bk)$ by Lemma (\ref{lemma:eigenvalue}). Additionally, as $\lambda^0_j(\bk)\to\infty$ as $j\to\infty$ we also know that $\omega^{+}_j\to\infty$ as $j\to\infty$. It only remains to show that $\omega^{-}_j\to \Omega$ as $j\to\infty$. As each $\omega^-_j(\bk)$ is bounded by $\Omega$ and the sequence of $\omega^-_j$ is increasing they must converge. Suppose that $\omega^-_j(k)\to \omega^- <\Omega$. Then there exists a $\phi$ such that
\begin{align}
     \left(-(\nabla+\iu\bk)^2\right)^{1/2}\phi +\frac{g^2\rho(\bx)}{\omega^--\Omega}\phi = \omega^-\phi.
\end{align}
This means that there is some $\lambda_i(\bk,\omega)$ such that $\lambda_i(\bk,\omega^-)=\omega^-$. However, if we take $\lambda_{i+1}(\bk,\omega)$ it must intersect the diagonal at a value larger than $\lambda_i(\bk,\omega)$, say $\omega_0$. Then $\omega_0$ is also an eigenvalue of $H(\bk)$ and $\omega_0 > \omega^-$ contradicting our assumption. \\

In order to see that the eigenvalues form a Hilbert basis for $L^2(\T^d;\C)$ it is sufficient to note that the spectrum (having only a single accumulation point at $\Omega$) is pure point and hence the spectral projection consists of only eigenvectors.
\end{proof}

\subsection{Band gap and lower bound of the bands} \label{sec:bounds}
In the case of constant density $\rho(\bx) = \rho_0$, we recall from \eqref{eq:bound_const} the following bounds on the band functions:
$$\frac{\Omega + \sqrt{\Omega^2 + 4g^2\rho_0 }}{2} \leq \omega^+_{m}(k), \qquad   \frac{\Omega - \sqrt{\Omega^2 + 4g^2\rho_0 }}{2} \leq  \omega^-_{m}(k).$$ 
The next two results show that these bounds can be generalized to arbitrary densities $\rho$. From this, we can conclude the existence of a band gap above $\Omega$, as well as a lower bound of the spectrum of $H$.
\begin{theorem}\label{thm:bandgap_general}
		Assume $\rho$ is $\Lambda$-periodic and strictly positive: $\rho(\bx) \geq m>0$. Then there is a spectral gap above $\Omega$: let $\mathcal{I} = \left(\Omega,\frac{\Omega + \sqrt{\Omega^2+4g^2m^2}}{2} \right)$, then
		\begin{equation}
			\sigma(H) \cap \mathcal{I} = \emptyset.
		\end{equation}
	\end{theorem}
\begin{proof}
Assuming $\rho(\bx) \geq m$, we have as in the proof of point (v) of \Cref{prop:properties}
\begin{align}
	\lambda_j(\bk,\omega) \geq \frac{gm}{\omega-\Omega},
\end{align}
for $\omega>\Omega$. In particular, this estimate is independent of $\bk$ and $j$. By \Cref{lemma:eigenvalue}, the spectrum of $H(\bk)$ is given by $\omega$ such that $\lambda_j(\bk,\omega) = \omega$ for some $j$. Here, the inequality $\omega \geq \frac{gm}{\omega-\Omega} $ is never satisfied for $\omega$ in the interval 
\begin{equation}
	\Omega <\omega < \frac{\Omega + \sqrt{\Omega^2+4g^2m^2}}{2},
\end{equation}
which finishes the proof.
\end{proof}
We conclude this section with a lower bound on the band structure.
\begin{prop}\label{prop:lower}
Assume $\|\rho\|_\infty < \infty$. We then have the following lower bound on the spectrum of $H$:
\begin{equation}
\inf \sigma(H) \geq \frac{\Omega - \sqrt{\Omega^2 + 4g^2\|\rho\|_{\infty} }}{2}.
\end{equation}
\end{prop}
The proof, which relies on properties of the pseudo-periodic Green's function, is given in \Cref{sec:propproof}.

\section{Crystals of high-contrast inclusions}\label{sec:highcontrast}
With the general theory of \Cref{sec:general} at hand, we proceed with a detailed study of the band functions in the case when $\rho(\bx)$ is piecewise constant and supported on a union $D$ of $N$ inclusions inside the fundamental cell as depicted in \Cref{fig:lattice}. Specifically, we will characterize the band functions in the case of \emph{high-contrast} inclusions, by which we mean that $\rho$ is supported on small inclusions with high density. We let $0<\epsilon \ll 1$ denote the length scale of the inclusions, and take $D=D_\epsilon$ to be a union of disjoint domains $D_{\epsilon,i} \subset Y$:
\begin{equation}
	D_\epsilon = \bigcup_{i=1}^N D_{\epsilon,i}.
\end{equation}
In the limit $\epsilon \to 0$, we take
\begin{equation}
	D_{\epsilon,i} = \epsilon B_i + \bz_i,
\end{equation}
for fixed domains $B_i\subset \R^d$ and (distinct) centre points $\bz_i\in Y$, independent of $\epsilon$. Then we take 
\begin{equation}
	\rho(\bx) = \sum_{\bm\in \Lambda} \rho_0(\epsilon) \chi_{D_\epsilon}(\bx-\bm),
\end{equation}
where $\chi_{D_\epsilon}$ is the characteristic function on $D_\epsilon$. Here, $\rho_0 = \rho_0(\epsilon)$ is an $\epsilon$-dependent constant, which we will assume is large when $\epsilon$ is small, corresponding to a high contrast. Depending on the spatial dimensionality $d$, we will take different scaling of $\rho_0(\epsilon)$ as $\epsilon \to 0$. 
\begin{figure}
\includegraphics[width=0.8\linewidth]{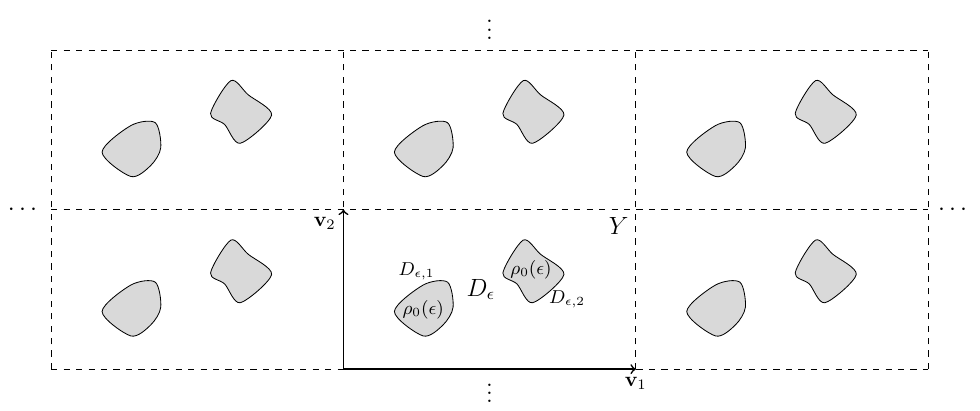}
\caption{Sketch of the periodic arrangement of high-contrast scatterers $D_\epsilon$ inside the unit cell $Y$, with $\epsilon$-dependent density $\rho_0(\epsilon)$. Here shown in the case of two scatterers ($N=2$) in two dimensions (d=2). }\label{fig:lattice}
\end{figure}

From \Cref{theorem_1}, we know that there are two families of band functions, $\omega^\pm(\bk)$, above and below $\Omega$, respectively. The goal of this section is to find analytic and numerical methods to compute these band functions in the current setting. As we will see, the two families emerge from rather different mathematical origins, briefly described as follows. In the limit $\epsilon \to 0$, the family $\omega^-_j(\bk)$ of band functions below $\Omega$ tend to the eigenvalues of a limiting operator, which is a translation of a compact operator and has a sequence of eigenvalues accumulating at $\Omega$. Using Gohberg-Sigal eigenvalue perturbation theory for holomorphic operator-valued functions, we are able to find asymptotic expansions of $\omega^-_j(\bk)$ of the form
\begin{equation} \label{eq:form}
\omega_{j}^-(\bk,\epsilon) \approx \omega_{0,j}^- + K_1(\epsilon)  + K_2(\bk,\epsilon),\end{equation}
where $K_1, K_2\to 0$ as $\epsilon \to 0$ and the constants $\omega_{0,j}^-$ are eigenvalues of the limiting operator. We present rigorous asymptotic expansions in \Cref{thm:band}, \Cref{thm:band2d}  and \Cref{thm:band1d} in the three-, two-, and one-dimensional cases, respectively.

Crucially, the expansions \eqref{eq:form} are not valid around points $(\omega,\bk)$ where $\omega = c|\bk+\bq|$ for some $\bq \in \Lambda^*$, known as  \emph{Rayleigh singularities}. In \Cref{sec:nonlin} we will present a formal  method to compute the band functions around these points. As we shall see, this will allow us to compute the family $\omega^+_j(\bk)$ of band functions above $\Omega$. 

\subsection{Pseudo-periodic Green's function and lattice sums}
We assume that $\omega \neq c|\bk+\bq|$ for all $\bq\in \Lambda^*$. Starting with the Green's function $G^{\omega/c}$ associated to $\left((-\Delta)^{1/2} - \frac{\omega}{c} \right)$, we can define the $\bk$-pseudo-periodic Green's function $G^{\bk,\omega/c}$ as 
\begin{equation}\label{eq:spatial}
	G^{\bk,\omega/c}(\bx) = \sum_{\bm \in \Lambda}e^{\iu \bk\cdot \bm} G^{\omega/c}(\bx-\bm).
\end{equation}
In \Cref{sec:G} we collect some results on $G^{\omega/c}$ from part I. In particular, for $\omega>0$ we have from \eqref{eq:Gid} that
\begin{equation}G^{\bk,\omega/c}(\bx) = \frac{2\omega}{c}\sum_{\bm \in \Lambda}e^{\iu \bk\cdot \bm} G^{\omega/c}_{\mathrm{helm}}(\bx-\bm) + \sum_{\bm \in \Lambda}e^{\iu \bk\cdot \bm} G^{-\omega/c}(\bx-\bm),\end{equation}
where $ G^{\omega/c}_{\mathrm{helm}}$ is the Green's function associated to the Helmholtz operator \begin{equation}\left(-\Delta - \frac{\omega}{c}\right).\end{equation} If $\omega \neq c|\bk+\bq|$ for all $\bq \in \Lambda^*$, the first sum is well-known to converge uniformly for $\bx$ in compact sets of $\R^d$, $\bx\neq 0$ (see \emph{e.g} \cite[Section 2.12]{ammari2018mathematical} or \cite{petit2013electromagnetic}). Moreover, from \Cref{sec:G} we know that $G^{-\omega/c}(\bx) = O\left(|\bx|^{-(d+1)}\right)$ as $|\bx|\to \infty$, so the second sum is uniformly and absolutely convergent for all $\bx\neq 0$. It is clear that
\begin{equation}\left((-\Delta)^{1/2} - \frac{\omega}{c} \right)G^{\bk,\omega/c}(\bx) =  \sum_{\bm\in\Lambda} e^{\iu \bk\cdot \bm}\delta(\bx- \bm).\end{equation}
From Poisson's summation formula we have that 
\begin{equation}\sum_{\bm\in\Lambda} e^{\iu \bk\cdot \bm}\delta(\bx- \bm) = \sum_{\bq\in \Lambda^*} e^{\iu (\bq+\bk) \cdot \bx}.\end{equation}
It therefore follows that $G^{\bk,\omega/c}$ admits a Fourier series as
\begin{equation}\label{eq:spectral}
G^{\bk,\omega/c}(\bx) = \sum_{\bq\in \Lambda^*} \frac{e^{\iu(\bq+\bk)\cdot \bx}}{|\bk+\bq| - \omega/c}.\end{equation}
As we shall see, the asymptotic behaviour of the band functions are described by the (zeroth order) \emph{lattice sum} $S^{\omega/c}(\bk)$, defined as
\begin{equation}\label{eq:S}
	S^{\omega/c}(\bk) = \sum_{\bm \in \Lambda\setminus\{0\} }e^{\iu \bk\cdot \bm} G^{\omega/c}(\bm).
\end{equation}
For $\omega>0$ we, can split $S^{\omega/c}$ into two components:
\begin{equation}S^{\omega/c}(\bk) = \frac{2\omega}{c}S_{\mathrm{helm}}^{\omega/c}(\bk) + S^{-\omega/c}(\bk),\end{equation}
where $S_{\mathrm{helm}}^{\omega/c}(\bk)$ is the lattice sum corresponding to the Helmholtz Green's function while $S^{-\omega/c}(\bk)$ is given by the absolutely convergent series coming from the $O\left(|\bx|^{-(d+1)}\right)$-remainder. 

\subsection{Integral equation formulation}
Recall that we are seeking $\omega = \omega(\bk)$ such that the problem
\begin{equation}\label{eq:psi2}
	\left\{
	\begin{array} {ll}
		\ds c(-\Delta)^{1/2}\psi  + \frac{g^2 \rho(\bx)}{\omega-\Omega}\psi = \omega \psi^\bk, \quad  & \bx \in Y, \\[0.8em]
		\psi(\bx+\bm) = e^{\iu \bk \cdot \bm}\psi(\bx), & \bm\in \Lambda,
	\end{array}\right.	
\end{equation}
has a nonzero solution $\psi$. Using the pseudo-periodic Green's function, we obtain from \eqref{eq:psi2} the equivalent integral equation
\begin{equation}
	\psi(\bx) =  -\frac{g^2\rho_0(\epsilon)}{c(\omega-\Omega)}\int_{D_\epsilon} G^{\bk,\omega/c}(\bx-\by) \psi(\by) \d \by.
\end{equation}
We define the integral operator $\G_\epsilon^{\bk,\omega}: L^2(D_\epsilon) \to L^2(D_\epsilon)$ as
\begin{equation}\label{eq:Akw}
	\G_\epsilon^{\bk,\omega}[\phi](\bx) = -(\omega-\Omega)\phi(\bx) -\frac{g^2\rho_0(\epsilon)}{c}\int_{D_\epsilon} G^{\bk,\omega/c}(\bx-\by) \phi(\by) \d \by.
\end{equation}
Consequently, the band functions $\omega = \omega(\bk)$ are characterized as solutions to the nonlinear eigenvalue problem
\begin{equation}
	\G_\epsilon^{\bk,\omega}[\phi](\bx) = 0.
\end{equation}

Since $D_\epsilon$ consists of several connected components, we will adopt a block matrix representation of $\G_\epsilon^{\bk,\omega}$. We let $\G^{\bk,\omega}_{\epsilon,(i,j)}: L^2(D_{\epsilon,j}) \to L^2(D_{\epsilon,i})$ be defined as in \eqref{eq:Akw} but integrated around $D_{\epsilon,j}$ and evaluated on $D_{\epsilon,i}$. We then have the following representation:
\begin{equation}\G^{\bk,\omega}_\epsilon = \begin{pmatrix} \G^{\bk,\omega}_{\epsilon,(1,1)} & \G^{\bk,\omega}_{\epsilon,(1,2)} & \cdots & \G^{\bk,\omega}_{\epsilon,(1,N)}\\ \G^{\bk,\omega}_{\epsilon,(2,1)} & \G^{\bk,\omega}_{\epsilon,(2,2)} & \cdots & \G^{\bk,\omega}_{\epsilon,(2,N)} \\ \vdots & \vdots & \ddots  & \vdots \\\G^{\bk,\omega}_{\epsilon,(N,1)} & \G^{\bk,\omega}_{\epsilon,(N,2)} & \cdots & \G^{\bk,\omega}_{\epsilon,(N,N)} \end{pmatrix}.\end{equation}
Here, we have identified $L^2(D_\epsilon) = L^2(D_{\epsilon,1})\times ... \times L^2(D_{\epsilon,N})$. We define $\L^2(B) =  L^2(B_1)\times ... \times L^2(B_N)$ and introduce the map $S_\epsilon$ between $\L^2(B) $ and $L^2(D_\epsilon)$:
\begin{equation}
	S_\epsilon: L^2(D_\epsilon) \rightarrow \L^2(B), \ (S_\epsilon f)_i(\bx) = f\bigl(\epsilon\bx+\bz_i).
\end{equation}

\subsection{Band structure in two and three dimensions}\label{sec:3d2s}
We begin with the case $d\in \{2,3\}$. In this case, we assume (as before) that
\begin{equation}
	D_{\epsilon,i} = \epsilon B_i + \bz_i, \quad D_\epsilon= \bigcup_{i=1}^N D_{\epsilon,i},
\end{equation}
and additionally assume that 
\begin{equation}
	\quad \rho(\bx) = \sum_{\bm\in \Lambda} \frac{s_0}{\epsilon}\chi_{D_\epsilon}(\bx-\bm),
\end{equation}
for some constant $s_0>0$. In this case, \textit{i.e.} when $\rho_0(\epsilon)$ scales as $O(\epsilon^{-1})$, it is known that the resonances of the individual inclusions scale as $O(1)$ part I.

When studying the case of small inclusions, we are interested in the expansion of the pseudo-periodic Green's function for $\bx$ close to $0$. As before, we assume that $\omega \neq c|\bk+\bq|$ for all $\bq \in \Lambda^*$, in which case we have from \Cref{sec:sing}
\begin{equation}\label{eq:Gexp}
	\epsilon^dG^{\bk,\omega/c}(\epsilon \bx) = \sum_{n=0}^\infty \epsilon^{n+1}A_n^{\bk,\omega/c}(\bx) +  \sum_{n=d-1}^\infty \epsilon^{n+1}\log(\epsilon)B_n^{\omega/c}(\bx),
\end{equation}
for functions $A_n^{\bk,k}$ and $B_n^k$ which can be explicitly computed (the first few terms, which we will subsequently need, are reported in \Cref{sec:sing}). For the translations of $G^{\bk,\omega/c}$, we also have expansions
\begin{equation}
	\epsilon^dG^{\bk,\omega/c}(\epsilon \bx + \bz_i - \bz_j) = \sum_{n=d-1}^\infty \epsilon^{n+1}A_{n,(i,j)}^{\bk,\omega/c}(\bx), \qquad i,j=1,...,N, \quad i\neq j,
\end{equation}
for functions $A_{n,(i,j)}^{\bk,\omega/c}(\bx)$. For ease of notation, we identify the diagonal terms with the $A$'s as in \eqref{eq:Gexp}: $A_{n,(i,i)}^{\bk,\omega/c}(\bx) = A_{n}^{\bk,\omega/c}(\bx)$. We then define two families of operators $\A_n^{\bk,\omega}, \B_n^{\omega}: \L^2(B) \to \L^2(B)$; each operator is block-wise defined and $\A_n^{\bk,\omega}$ is given by
\begin{equation}
	\A^{\bk,\omega}_{n,(i,j)}[\phi_j](\bx) = -\frac{g^2s_0}{c}\int_{B_j} A^{\bk,\omega/c}_{n,(i,j)}(\bx-\by)\phi_j(\by) \d \by, \quad \bx \in B_i,
\end{equation}
for $i,j = 1,...,N$ and $n>0$. For $n=0$, we define $\A^{\bk,\omega}_{0}$ as the block-wise diagonal operator given by
\begin{equation}
	\A^{\omega}_{0,(i,j)}[\phi_j](\bx) = \delta_{ij}\left(-(\omega-\Omega)\phi(\bx) -\frac{g^2s_0}{c}\int_{B_j} A_0(\bx-\by)\phi_j(\by) \d \by \right), \quad \bx \in B_i,
\end{equation}
In particular, $\A^{\bk,\omega}_0 = \A^{\omega}_0$ is independent of $\bk$, and hence we omit this superscript. We also define the block-wise diagonal operators $\B_{n}^{\omega} : \L^2(B) \to \L^2(B)$ as
\begin{equation}
\B^{\omega}_{n,(i,j)}[\phi_j](\bx) =  \delta_{ij}\left(-\frac{g^2s_0}{c}\int_{B_j} B^{\omega/c}_{n}(\bx-\by)\phi_j(\by) \d \by\right), \quad \bx \in B_i,
\end{equation}
We then have the following asymptotic expansion of $\G^{\bk,\omega}_\epsilon$.
\begin{prop}
	Assume that $\omega \neq  c|\bk +\bq|$ for all $\bq \in \Lambda^*$. We then have 
	\begin{equation}\label{eq:expA}
		S_\epsilon\G^{\bk,\omega}_\epsilon S_\epsilon^{-1} = \sum_{n=0}^\infty \epsilon^n\A^{\bk,\omega}_n + \sum_{n=d-1}^\infty \epsilon^n\log(\epsilon)\B^{\omega}_{n},\end{equation}
	where the convergence holds in $\B\bigl(\L^2(B)\bigr)$.
\end{prop}

When $\omega \notin \{0,\Omega\}$ and $\omega \neq  c|\bk +\bq|$ for all $\bq \in \Lambda^*$, $\G_\epsilon^{\bk, \omega}$ is a holomorphic operator-valued function of $\omega$, and we can follow a similar approach as in part I to derive asymptotic expansions of eigenvalues of $\G_\epsilon^{\bk,\omega}$. The eigenvalue problem for the limiting operator $\A^\omega_0$ is a linear eigenvalue problem; $\A^\omega_0[\psi] = 0$ is equivalent to 
\begin{equation}L_0\psi = (\Omega-\omega) \psi,\end{equation}
where $L_0$ is block-wise diagonal and defined by
\begin{equation}\label{eq:L0}
	L_{0,(i,j)}[\phi_j](\bx) = \delta_{ij}\left(\frac{g^2s_0}{c}\int_{B_j} A_0(\bx-\by)\phi_j(\by) \d \by \right), \quad \bx \in B_i.
\end{equation}
Here $L_0$ is a compact operator, and has a sequence of eigenvalues $\lambda = \Omega - \omega_{0,j}^-$ converging to zero as $j \to \infty$. By positivity of $(-\Delta)^{1/2}$ we know that $L_0$ is a positive operator, so we have $\omega_{0,j}^- < \Omega$. Starting with these eigenvalues $\omega_{0,j}^-$, we can use Gohberg-Sigal theory for perturbations of holomorphic operator-valued functions to compute asymptotic expansions of eigenvalues $\omega_j^-$ of $\G_\epsilon^{\bk,\omega}$ \cite{ammari2018mathematical,Gohberg1971}. The first result is analogous to \cite[Proposition 3.6]{one-photon_bound}.
\begin{prop} \label{prop:pertQP}
	Let $\omega_{0,j}^-\in \R\setminus\{0,\Omega\}$ be an eigenvalue of $\A^{\omega}_{0}$ of multiplicity $M$ and let $\bk$ be fixed such that $\omega_{0,j}^- \neq c|\bk+\bq|$ for all $\bq \in \Lambda^*$. Then, for small enough $\epsilon$, there exists $M$ eigenvalues (up to multiplicity) $\omega_{j}^- = \omega_{j}^-(\epsilon), \ j=1,...,M$ of $\G_\epsilon^{\bk,\omega}$, continuous as a functions of $\epsilon$ and satisfying $\lim_{\epsilon \to 0^+}\omega_{j}^-(\epsilon) = \omega_{0,j}^-$.
\end{prop}
In the following, we assume that $D_\epsilon$ consists of two identical inclusions with shape $B \subset \R^d$:
\begin{equation}D_\epsilon= D_{\epsilon,1} \cup D_{\epsilon,2}, \qquad D_{\epsilon,i} = \epsilon R_i B + \bz_i,\end{equation}
for some rotations $R_i$ and centres $\bz_i$. The general case of $N$ inclusions can be treated using the method developed in \cite{ammari2004splitting}. In the case of two inclusions,  $\A_0^\omega$ can be represented by a $2\times 2$ diagonal block matrix, and every eigenvalue $\omega_{0,j}^-$ of $\A_0^\omega$ has even multiplicity. If we have a double eigenvalue $\omega_{0,j}^-$, it will split into two eigenvalues $\omega_{j}^-$ and $\omega_{j+1}^-$ of the full operator $\G^{\bk,\omega}_\epsilon$ when $\epsilon$ is nonzero. The next theorem is the main result in $d=3$.
\begin{theorem}\label{thm:band}
	Assume $d=3$ and let $\omega_{0,j}^-\in \R\setminus\{0,\Omega\}$ be a double eigenvalue of $\G^{\omega}_0$ with  corresponding normalized eigenmodes $\psi_1 = (\psi_{(1,1)}, 0)$, $\psi_2 = (0,\psi_{(2,2)})$. Moreover, let $\bk$ be fixed such that $\omega_{0,j}^- \neq c|\bk+\bq|$ for all $\bq \in \Lambda^*$.  Then, for small enough $\epsilon$, there are two eigenvalues $\omega_{j+i}^- = \omega_{j+i}^-(\bk,\epsilon), \ i=0,1,$ of $\G^{\bk,\omega}_\epsilon,$ satisfying
	\begin{multline} \label{eq:thm3d}
		\omega_{j+i}^-(\bk,\epsilon) = \omega_{0,j}^- + \epsilon\left\langle \psi_{1},\A^{\omega_{0,j}^-}_{1}\psi_{1}\right\rangle  + \epsilon^2\left( C_1 + C_2\log(\epsilon) \right)  \\ -  \epsilon^2\frac{g^2s_0}{c}\left(\int_{B_1} \psi_{(1,1)}(\bx)\d \bx \right)^2 \left( S^{\omega_{0,j}^-}(\bk) +(-1)^i \left| G^{\bk,\omega_{0,j}^-/c}(\bz_1-\bz_2)\right| \right)  + O(\epsilon^3\log\epsilon),
	\end{multline}
	for constants $C_1$ and $C_2$ which are independent of $\epsilon$ and $\bk$.
\end{theorem}
\begin{proof}
We assume $\omega_{0,j}^-\in \R$ is a double eigenvalue of $\G^{\omega}_{0}$ and let $V\subset \C$ be a complex neighbourhood containing no other eigenvalues of $\G^{\omega}_{0}$. We define the coefficients $a_1, a_2$ as
\begin{equation}
	a_l = \frac{1}{2\pi \iu}\tr\int_{\partial V} (\omega-\omega_{0,j}^-)^l\left(\G^{\bk,\omega}_\epsilon\right)^{-1} \frac{\d }{\d \omega}\G^{\bk,\omega}_\epsilon \d \omega, \quad l = 1,2.
\end{equation}
By the generalized argument principle (see, e.g. \cite[Theorem 1.14]{ammari2018mathematical} and \cite[Appendix D]{one-photon_bound}) , we know that $\omega_{j+i}^-$ satisfy
\begin{align}
	 (\omega_{j}^--\omega_{0,j}^-) + (\omega_{j+1}^- -\omega_{0,j}^-) &= a_1 \\ 
	 (\omega_{j}^--\omega_{0,j}^-)^2 + (\omega_{j+1}^- -\omega_{0,j}^-)^2 &= a_2,
\end{align}
for $i=0,1$, so that 
\begin{equation}\label{eq:omega12}
\omega_{j+i}^- -\omega_{0,j}^- = \frac{a_1 + (-1)^i \sqrt{2a_2 -a_1^2}}{2}.
\end{equation}
We have \cite[Proposition 7.2]{ammari2004splitting}
\begin{equation}\label{eq:a}
	a_l = \frac{1}{2\pi \iu}\sum_{p=1}^\infty\frac{l}{p}
\tr\int_{\partial V}(\omega-\omega_{0,j}^-)^{l-1}\left[\left(\A_0^\omega\right)^{-1}\left(\A_0^\omega - \G^{\bk,\omega}_\epsilon \right)\right]^p\d \omega,
\end{equation}
where, for $\omega$ in a neighbourhood of $\omega_{0,j}^-$,
\begin{equation}\left(\A_0^\omega \right)^{-1} = -\frac{\L}{(\omega-\omega_{0,j}^-)} + \sum_{n=0}^\infty (\omega-\omega_{0,j}^-)^nL_n.
\end{equation}
Here, $\L$ is an operator from the kernel of $\A_0^{\omega_{0,j}^-}$ onto itself, while $L_n,$ for $n\geq 0,$ vanishes on the kernel of $\A_0^{\omega_{0,j}^-}$. We pick a basis $\{\psi_1, \psi_2\}$ of $\ker\left(\A_0^{\omega_{0,j}^-}\right)$ as $\psi_1 = \left(\psi_{(1,1)},0\right), \psi_2 = \left(0, \psi_{(2,2)}\right)$, where $\psi_{(i,i)}$ is a normalized eigenmode of $\A^{\omega_{0,j}^-}_{0,(i,i)}$. In this basis, we have matrix representations 
\begin{equation}
	\L\E = \begin{pmatrix} e_1(\omega) & e_2(\omega) \\ {e_2}^*(\omega) & e_1(\omega) \end{pmatrix}, \quad \L\E L_0\E = \begin{pmatrix} l_1(\omega) & l_2(\omega) \\ {l_2}^*(\omega) & l_1(\omega) \end{pmatrix},
\end{equation}
where $\E = \bigl( \G^{\bk,\omega}_\epsilon - \A_0^{\omega}\bigr)$. A straight-forward calculation shows that, as $\epsilon \to 0$,
\begin{equation}\label{eq:exp}
	\frac{a_1 \pm \sqrt{2a_2 -a_1^2}}{2} = e_1(\omega_{0,j}^-) + \frac{1}{2}\frac{\p}{\p\omega}e_1(\omega_{0,j}^-)^2 + l_1(\omega_{0,j}^-) \pm |e_2(\omega_{0,j}^-)|+ O(\epsilon^3\log\epsilon).
\end{equation}
In the three-dimensional case, $\A^{\bk,\omega}_{2}$ is the leading term which depends on $\bk$:
\begin{equation}\A^{\bk,\omega}_{2,(i,j)}[\phi](\bx) = \begin{cases} \ds - \frac{g^2s_0}{c}\int_{B_j} \left(A^{\omega/c}_2(\bx-\by) + S^{\omega/c}(\bk) \right) \phi(\by)\d \by ,\quad &i=j, \\
	\ds - \frac{g^2s_0}{c} G^{\bk,\omega/c}(\bz_i-\bz_j)\int_{B_j} \phi(\by) \d \by, & i\neq j.
\end{cases}
\end{equation}
Observe that $e_1$ and $e_2$ are given by
\begin{equation}
	e_1(\omega_{0,j}^-) = \left\langle \psi_{1}, \bigl( \G_\epsilon^{\bk,\omega_{0,j}^-} - \A_{0}^{\omega_{0,j}^-}\bigr)\psi_{1}\right\rangle, \quad e_2(\omega_{0,j}^-) = \left\langle \psi_{1}, \bigl( \G_\epsilon^{\bk,\omega_{0,j}^-} - \A_{0}^{\omega_{0,j}^-}\bigr)\psi_{2}\right\rangle. 
\end{equation}
From \eqref{eq:expA} we then have
\begin{equation}
	e_1(\omega_{0,j}^-) = \epsilon\left\langle \psi_{1},\A^{\omega_{0,j}^-}_{1}\psi_{1}\right\rangle + \epsilon^2\left(  \tilde{C}_1 + \tilde{C}_2\log\epsilon\right) -  \epsilon^2\frac{g^2s_0}{c}U^2 S^{\omega_{0,j}^-}(\bk) + O(\epsilon^3\log\epsilon),
\end{equation}
where $U = \int_{B_1} \psi_{(1,1)}(\bx)\d \bx$ and $\tilde{C}_i$ denote constants with respect to $\epsilon$ and $\bk$. Furthermore, we have
\begin{equation}
	\frac{\p}{\p\omega}e_1(\omega_{0,j}^-)^2 = \epsilon^2\left(  \tilde{C}_3 + \tilde{C}_4\log\epsilon\right) + O(\epsilon^3\log\epsilon), \quad  l_1(\omega_{0,j}^-) = \epsilon^2\left(  \tilde{C}_5 + \tilde{C}_6\log\epsilon\right) + O(\epsilon^3\log\epsilon),
\end{equation}
and, finally,
\begin{equation}
	|e_2(\omega_{0,j}^-)| = \epsilon^2 \frac{g^2s_0}{c}\left|G^{\bk,\omega_{0,j}^-/c}(\bz_i-\bz_j)\right|U^2 +O(\epsilon^3\log\epsilon).
\end{equation}
Combining these expansions with \eqref{eq:exp} and \eqref{eq:omega12} proves the result.
\end{proof}

\begin{remark}
	The constants $C_1$ and $C_2$ can be explicitly computed. However, for the purpose of band function calculations, we omit this calculation and focus on the terms which depend on $\bk$.
\end{remark}

\begin{remark}\label{rmk:N13d}
The term $|e_2(\omega_{0,j}^-)|$ describes the hybridization between the inclusions. If $|e_2(\omega_{0,j}^-)| \neq 0$, the hybridization causes the eigenvalue $\omega_{0,j}^-$ to split into two eigenvalues $\omega_{j}^-$ and $\omega_{j+1}^-$. In the case of a single inclusion inside the unit cell ($N=1$), this term is not present, and we can show that there is a single eigenvalue $\omega_{j}^- = \omega_{j}^-(\bk, \epsilon)$ in a neighbourhood of $\omega_{0,j}^-$, which satisfies
\begin{multline}
	\omega_{j}^-(\bk,\epsilon) = \omega_{0,j}^- + \epsilon\langle \psi_1, \A_1^{\omega_{0,j}^-} \psi_1\rangle + \epsilon^2\left( C_1 + C_2\log(\epsilon) \right) -  \epsilon^2\frac{g^2s_0}{c}\left(\int_{B_1} \psi_{(1,1)}(\bx)\d \bx \right)^2  S^{\omega_{0,j}^-}(\bk) \\ + O(\epsilon^3\log\epsilon).
\end{multline}

\end{remark}

\begin{remark} \label{rmk:highbands}
	As pointed out in part I, there is a sequence of eigenvalues $\omega_{0,j}^-$ of $\A_0^\omega$ converging to $\Omega$ as $j\to \infty$ and such that $\omega_{0,j}^- \leq \Omega$. \Cref{thm:band} describes how these eigenvalues give rise to a sequence of band functions $\omega_{j}^-(\bk)$, accumulating at $\Omega$. In the periodic crystal, we will, in addition, have band functions $\omega_j^+(\bk)$ above $\Omega$. These bands are not captured by \Cref{thm:band}, and as $\epsilon \to 0$, any such band tends to the light cone $\omega = c|\bk + \bq|$ for some $\bq \in \Lambda^*$. In \Cref{sec:nonlin} we develop a formal method to compute such bands.
\end{remark}

Next, we consider the two-dimensional case. The leading term which depends on $\bk$ is now $\A^{\bk,\omega}_{1}$:
\begin{equation}
	\A^{\bk,\omega}_{1,(i,j)}[\phi](\bx) = \begin{cases} \ds - \frac{g^2s_0}{c}\int_{B_i} \left(A^{\omega/c}_1(\bx-\by) + S^{\omega/c}(\bk) \right) \phi(\by)\d \by ,\quad &i=j, \\
		\ds - \frac{g^2s_0}{c} G^{\bk,\omega/c}(\bz_i-\bz_j)\int_{B_j} \phi(\by) \d \by, & i\neq j.
	\end{cases}
\end{equation}
Observe that \eqref{eq:exp} and \eqref{eq:omega12} are not specific to the dimensionality, and are also valid in $d=2$. We then obtain the following result.
\begin{theorem}\label{thm:band2d}
	Assume $d=2$ and let $\omega_{0,j}^-\in \R\setminus\{0,\Omega\}$ be a double eigenvalue of $\A^{\omega}_0$ with corresponding normalized eigenmodes $\psi_1 = (\psi_{(1,1)}, 0)$, $\psi_2 = (0,\psi_{(2,2)})$. Assume that $\bk$ is fixed and such that $\omega_{0,j}^- \neq c|\bk+\bq|$ for all $\bq \in \Lambda^*$.  Then, for small enough $\epsilon$, there are two eigenvalues $\omega_{j+i}^- = \omega_{j+i}^-(\bk,\epsilon), i=0,1$, of $\G^{\bk,\omega}_\epsilon,$ satisfying
	\begin{multline}
		\omega_{j+i}^-(\bk,\epsilon) = \omega_{0,j}^- + \epsilon \left(C_1 + C_2\log \epsilon \right)\\ -  \epsilon\frac{g^2s_0}{c}\left(\int_{B_1} \psi_{(1,1)}(\bx)\d \bx \right)^2 \left( S^{\omega_{0,j}^-}(\bk) + (-1)^i \left| G^{\bk,\omega_{0,j}^-/c}(\bz_1-\bz_2)\right| \right)  + O(\epsilon^2\log\epsilon),
	\end{multline}
	for some constants $C_1, C_2$ which are independent of $\epsilon$, $i$  and $\bk$.
\end{theorem}

\subsection{One-dimensional crystals}
In the case of a single spatial dimension, we follow part I and take 
\begin{equation}\label{eq:rho01d}
	\rho_0(\epsilon) = -\frac{s_0}{\epsilon\log\epsilon}.
\end{equation}
Similarly to the previous section, this choice of scaling of $\rho_0$ will lead to band functions of order $O(1)$ as $\epsilon \to 0$. Then we have the asymptotic expansion \begin{equation}S_\epsilon\G_\epsilon^{k,\omega}S_\epsilon^{-1} = \sum_{n=0}^\infty \frac{\epsilon^n}{\log(\epsilon)}\A^{k,\omega}_n + \sum_{n=0}^\infty \epsilon^n\B^{\omega}_{n},\end{equation}
convergent in $\B(\L,\L)$. Here, $\B^{k,\omega}_{n}$ is independent of $k$ (and hence we omit this superscript) and is block-wise diagonal: 
\begin{align} \label{eq:A01d}
	\B^{\omega}_{0,(i,i)}[\phi](x) &= -(\omega-\Omega)\phi(x) -\frac{g^2s_0}{\pi c}\int_{B_i}\phi(y) \d y,
\end{align}
while $\A^{k,\omega}_0$ depends on $k$ and describes the interactions between the inclusions: 
\begin{equation}
	\A^{k,\omega}_{0,(i,j)}[\phi](x) = \begin{cases} \ds \frac{g^2s_0}{c}\int_{B_i} \left(A^{\omega/c}_0(x-y) + S^{\omega/c}(p) \right) \phi(y)\d y ,\quad &i=j, \\
	\ds \frac{g^2s_0}{c} G^{k,\omega/c}(z_i-z_j)\int_{B_j} \phi(y) \d y, & i\neq j.
\end{cases}
\end{equation}
As before, we assume that $N=2$ so that $D_\epsilon$ is a union of two intervals. In the one-dimensional case, we then have the following result.
\begin{theorem}\label{thm:band1d}
	Let $d=1$ and $N=2$. Assume that $\omega_{0,0}^-:=\Omega - \frac{g^2s_0|B_1|}{\pi c}\neq 0$ and let $k$ be fixed such that $\omega_{0}^- \neq c|k+q|$ for all $q \in \Lambda^*$. Then, for small enough $\epsilon$, there are two eigenvalues $\omega_{i}^- = \omega_{i}^-(k,\epsilon), i=0,1$, of $\G^{k,\omega}_\epsilon$, satisfying
	\begin{align}\label{eq:1d}
		\omega_{i}^-(\bk,\epsilon) &= \Omega - \frac{g^2s_0|B_1|}{\pi c} + \frac{C}{\log(\epsilon)} + \frac{g^2s_0|B_1|}{c\log(\epsilon)} \left( S^{\omega_{0}^-}(k) + (-1)^i \left| G^{k,\omega_{0,0}^-/c}(z_1-z_2)\right| \right)  + O(\epsilon),
	\end{align}
	where the constant $C = C(\omega_{0}^-)$ is independent of $\epsilon$, $i$ and $k$ and given by
	\begin{equation}
		C(\omega_{0,0}^-)= -\frac{g^2s_0|B_1|}{\pi c} \left(\gamma -\frac{3}{2} + \log\left(\frac{|\omega_{0,0}^-||B_1|}{c}\right)\right).
	\end{equation} 
\end{theorem}
\begin{remark}
	With the scaling of $\rho_0$ as in \eqref{eq:rho01d} and with $N=2$, there are two band functions $\omega_0^-$ and $\omega_1^-$ captured by \eqref{thm:band1d}; all other band functions $\omega_j^-$ will tend to $\Omega$ as $\epsilon \to 0$.
\end{remark}

\begin{remark}
	As in \Cref{rmk:N13d}, we can phrase \Cref{thm:band1d} for the case of a single inclusion ($N=1$): in this case there is a single eigenvalue $\omega_{0}^- = \omega_{0}^-(k,\epsilon)$ satisfying 
	\begin{equation}\label{eq:1dN1}
		\omega_{0}^-(k,\epsilon) = \Omega - \frac{g^2s_0|B_1|}{\pi c} + \frac{C}{\log(\epsilon)} + \frac{g^2s_0|B_1|}{c\log(\epsilon)} S^{\omega_{0,0}^-}(p)+ O(\epsilon).
	\end{equation}
\end{remark}

	\subsection{Negative bands}\label{sec:neg}
	As shown in part I, some of the eigenvalues $\omega_{0,j}^-$ could happen to be negative. Remarkably, in this case, there are no Rayleigh singularities, and the corresponding asymptotic formulas are valid uniformly in $\bk$. We phrase the result in three dimensions; the other two cases follow analogously.
	\begin{cor}\label{thm:bandned3D}
		Assume $d=3$ and let $\omega_{0,j}^- <0$ be a double eigenvalue of $\A^{\omega}_0$ with  corresponding normalized eigenmodqes $\psi_1 = (\psi_{(1,1)}, 0)$, $\psi_2 = (0,\psi_{(2,2)})$. Then, for small enough $\epsilon$, there are two eigenvalues $\omega_{j+i}^- = \omega_{j+i}^-(\bk,\epsilon), i=0,1,$ of $\G^{\bk,\omega}_\epsilon,$ satisfying
		\begin{multline}
			\omega_{j+i}^-(\bk,\epsilon) = \omega_{0,j}^- + \epsilon\left\langle \psi_{1},\A^{\omega_{0,j}^-}_{1}\psi_{1}\right\rangle  + \epsilon^2\left( C_1 + C_2\log(\epsilon) \right)  \\ -  \epsilon^2\frac{g^2s_0}{c}\left(\int_{B_1} \psi_{(1,1)}(\bx)\d \bx \right)^2 \left( S^{\omega_{0,j}^-}(\bk) +(-1)^i \left| G^{\bk,\omega_{0,j}^-/c}(\bz_1-\bz_2)\right| \right)  + O(\epsilon^3\log\epsilon),
		\end{multline}
		valid uniformly for $\bk \in \T^d$, where the constants $C_1$ and $C_2$ are independent of $\epsilon$ and $\bk$.
	\end{cor}
\begin{proof}
	Since $\omega_{0,j}^- < 0$, the condition $\omega_{0,j}^- \neq c|\bk+\bq|$ holds for all $\bk \in \T^d$. Since the asymptotic formula \eqref{eq:thm3d} is valid point-wise for all $\bk$ in the compact domain $\T^d$, it follows that it is valid uniformly in $\bk$.
\end{proof}

	The next corollary follows from the fact that the spectrum of $\A_0^\omega$ is discrete, so that bands corresponding to distinct eigenvalues will be well-separated.
	\begin{cor}\label{thm:bandgap}
	Under the assumptions of \Cref{thm:bandned3D}, and for $\epsilon$ small enough, there are band gaps below $\omega_j^-$ and above $\omega_{j+1}^-$. If, additionally, $\left| G^{\bk,\omega_{0,j}^-/c}(\bz_1-\bz_2)\right|\neq 0$ for all $\bk \in \T^d$, there is a band gap between $\omega_{j}^-$ and $\omega_{j+1}^-$ for small enough $\epsilon$.
	\end{cor}

\begin{remark}
	In the case of a negative band function, the corresponding Bloch modes are linear combinations of bound states of each individual inclusion. In this situation, the asymptotic analysis of this section can be understood as a tight-binding approximation. In the case of positive band functions, the eigenmodes are no longer bound states. Even though there is no tight-binding approximation in this case, the asymptotic analysis of \Cref{sec:3d2s} can still be used to compute the band functions.
\end{remark}
In \Cref{fig:bandsNeg} we illustrate a case in $d=1$ where there are negative band functions. Compared to \Cref{fig:bandsb}, we see that there is no singularity of the asymptotic formula.

\begin{figure}[htb]
	\centering
	\includegraphics[width=0.5\linewidth]{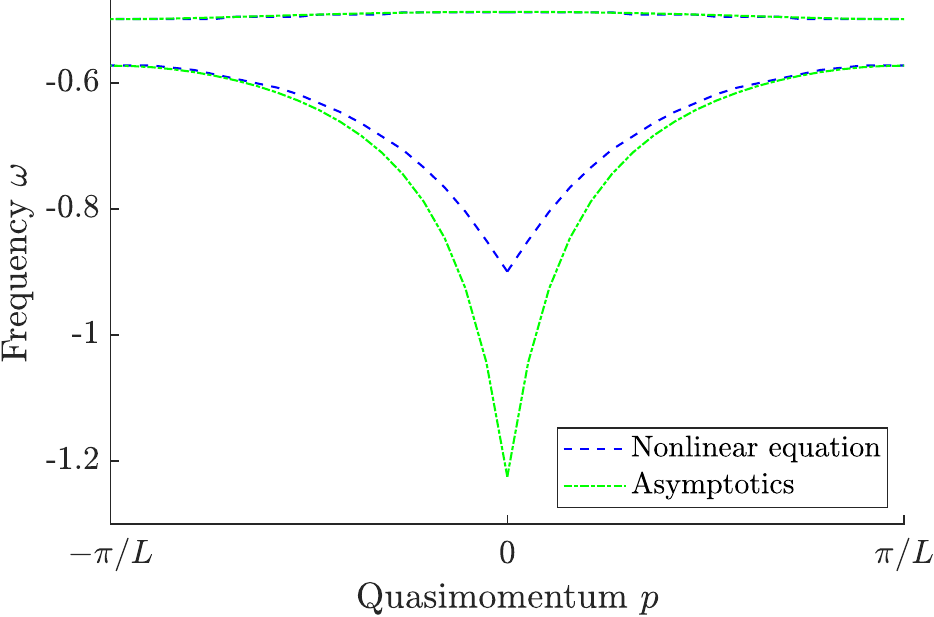}
	\caption{Band functions in $d=1$ and in the case of a pair of inclusions ($N=2$) of radius $R=10^{-4}$ with $\Omega = 0.1$. With these parameters, $\omega_{0,0}^- $ is negative so we have two band functions with negative frequencies. In this case, the asymptotic formula \Cref{fig:bandsb} does not have a Rayleigh singularity, and is valid for all $\bk \in \T$.}\label{fig:bandsNeg}
\end{figure}

\section{Rayleigh singularity and lattice sum regularization}\label{sec:nonlin}
	At the Raighleigh singularities $(\omega,\bk)$ where $\omega = c|\bk+\bq|$ for some $\bq\in \Lambda^*$, the lattice sum $S^{\omega/c}(\bk)$ diverges. The asymptotic expansions given in \Cref{thm:band}, \Cref{thm:band2d} and \Cref{thm:band1d} are in general not valid uniformly $\bk$, and will break around these singularities. The goal of this section will be to develop a formal approach to overcome this singularity and find the band functions close to the Rayleigh singluarities.
	
	\begin{figure}[htb]
		\centering
		\begin{subfigure}[t]{0.32\linewidth}
			\includegraphics[width=\linewidth]{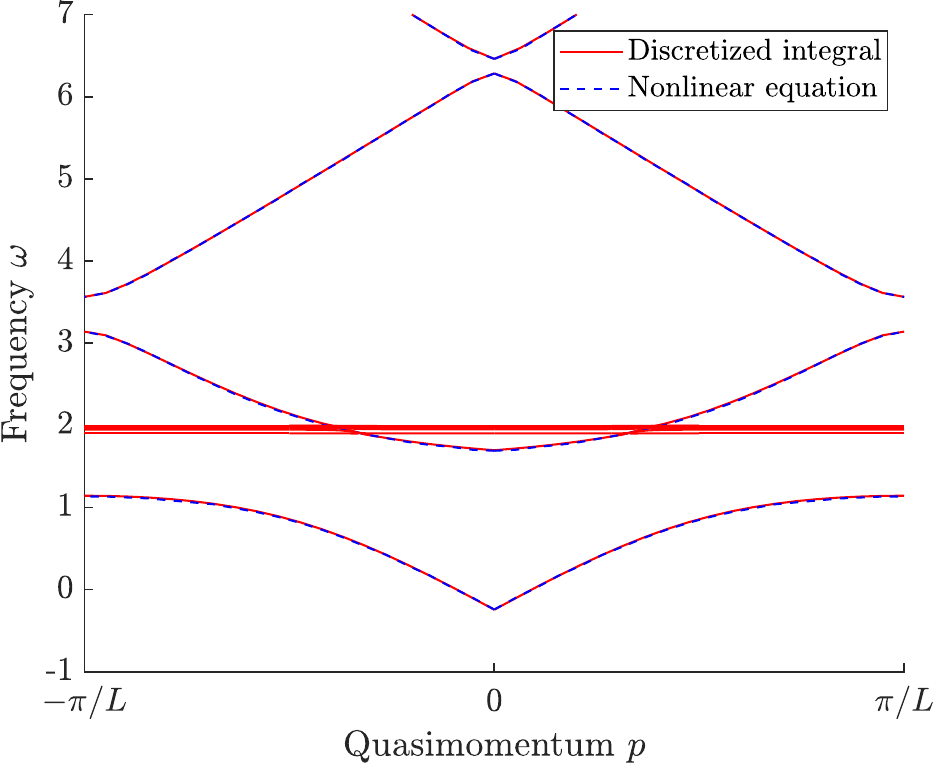}
			\caption{Band structure.}\label{fig:bandsa}
		\end{subfigure}\hfill
		\begin{subfigure}[t]{0.32\linewidth}
			\includegraphics[width=\linewidth]{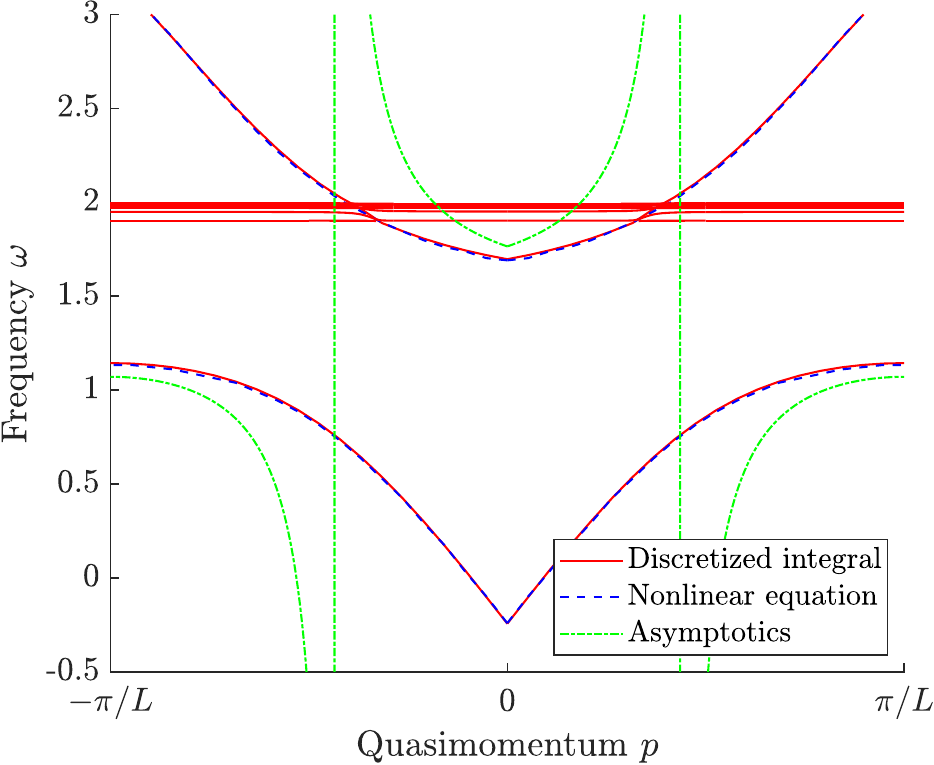}
			\caption{Comparison with the asymptotic formula \eqref{eq:1dN1}.}\label{fig:bandsb}
		\end{subfigure}\hfill	
		\begin{subfigure}[t]{0.32\linewidth}
			\includegraphics[width=\linewidth]{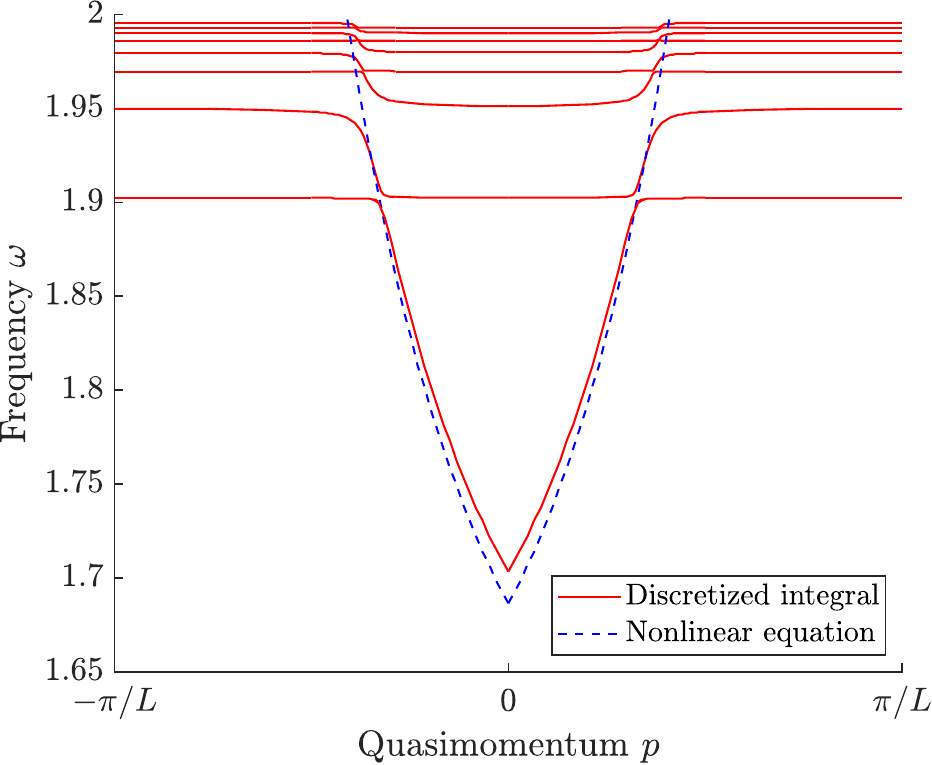}
			\caption{Close-up showing the bands which accumulate at $\Omega$.}\label{fig:bandsc}
		\end{subfigure}
		\caption{Band functions in $d=1$ and in the case of a single inclusion ($N=1$) of radius $R=0.01$ with $\Omega = 2$. Here, we compare results obtained by discretizing the integral equation \eqref{eq:Akw} with the results obtained by solving the nonlinear equation \eqref{eq:1dnonlin}. As seen in \Cref{fig:bandsa} and \Cref{fig:bandsc}, the nonlinear equation approximates well the first and the higher bands, but fails to  capture the accumulation of bands at $\Omega$. \Cref{fig:bandsb} shows that the asymptotic formula \eqref{eq:1d} diverges at the Rayleigh singularity $\omega = c|k|$, but provides an asymptotic approximation away from these points. Here, and throughout this work, we use $c=1$, $L=1$, $s_0=1$ and $g=1$.} \label{fig:bands}
	\end{figure}

	\subsection{Lattice sum regularization method in one dimension}\label{sec:lattice1d}
	The main approach to overcome the Rayleigh singularity is by formally phrasing the asymptotic expansions as nonlinear equations. We exemplify the approach in the case $d=1$ and $N=1$. In place of \eqref{eq:1dN1}, we approximate the eigenvalues $\omega = \omega_{j}^\pm$ by the roots of the nonlinear equation
	\begin{align}\label{eq:1dnonlin}
		\omega &= \Omega - \frac{g^2s_0|B_1|}{\pi c} + \frac{C(\omega)}{\log(\epsilon)} + \frac{g^2s_0|B_1|}{c\log(\epsilon)} S^{\omega}(k).
	\end{align}
In \Cref{fig:bands} we compare numerically computed values, computed asymptotically through \eqref{eq:1dN1} and \eqref{eq:1dnonlin}, as well as through discretization of the original integral equation. In \Cref{fig:bandsb} we see that there are singular points where \eqref{eq:1dN1} breaks down, whereas \eqref{eq:1dnonlin} is not affected by these singularities. Additionally, the first and the higher bands can be approximated remarkably well through this approach.

In the limit $\epsilon \to 0$, we can further simplify the solutions to the nonlinear equation \eqref{eq:1dnonlin}. We proceed formally. There are two types of solutions: either $\omega = \omega_j^-$ is away from the Rayleigh singularity and we have
\begin{equation}\omega_j^- \approx \Omega - \frac{g^2s_0|B_1|}{\pi c},\end{equation}
as $\epsilon \to 0$. Since the lattice sum $S^{\omega/c}(k)$ is unbounded for $\omega$ close to the light cone, there is a second type of solution $\omega = \omega_j^\pm(k)$ of the form
\begin{equation}\omega_j^\pm(k) = c|k+q_j| + \alpha,\end{equation}
for some $q_j \in \Lambda^*$ and some $\alpha \ll 1$. Based on \eqref{eq:Sdecomp} and  \eqref{eq:Sh}, we have a singularity of $S^{\omega/c}_\mathrm{helm}$ while $S^{-\omega/c}$ is bounded. We find from \eqref{eq:Sh} and \eqref{eq:1dnonlin} that $\alpha$ has the form
\begin{equation}
	\alpha \approx -\frac{\alpha_0}{\log\epsilon},
\end{equation}
for some constant $\alpha_0$. We then have 
\begin{equation}
	\cos\left(\frac{\omega L}{c}\right) - \cos(kL) \approx \frac{\alpha_0}{\log\epsilon}\sin(|k+q_j|L).
\end{equation}
Since $C(\omega)$ is bounded for $\omega$ in a neighbourhood of the light cone, we find from \eqref{eq:1dnonlin} that 
\begin{equation}
	\Omega - \frac{g^2s_0|B_1|}{\pi c} + \frac{g^2s_0|B_1|}{c\alpha_0} - c|k+q_j| \approx 0.
\end{equation}
Solving for $\alpha_0$, we find that 
\begin{equation}
	\alpha_0 =  \frac{g^2s_0|B_1|}{c\left(c|k+q_j| - \left(\Omega - \frac{g^2s_0|B_1|}{\pi c}\right)\right)}.
\end{equation}
In total, we approximate the band function $\omega_j^\pm$ by
\begin{equation}\label{eq:approx}
	\omega_j^\pm(k) -c|k+q_j| \approx -   \frac{g^2s_0|B_1|}{c\log\epsilon \left(c|k+q_j| - \left(\Omega - \frac{g^2s_0|B_1|}{\pi c}\right)\right)}.
\end{equation}
In \Cref{fig:nonlin} we numerically illustrate the approximation \eqref{eq:approx} at a fixed $k$. For larger $j$, and hence larger $q_j$, the coefficient $\alpha_0$ decreases and the band functions approach the light cone. For the first band, there is a large discrepancy between the asymptotic approximation \eqref{eq:approx} and the nonlinear equation \eqref{eq:1dnonlin}, but the higher bands are well approximated by \eqref{eq:approx}. In these crude asymptotics we only account for corrections of order $(\log\epsilon)^{-1}$, and more accurate approximations can be obtained with higher-order asymptotics.

\begin{figure}[htb]
	\centering
		\includegraphics[width=0.6\linewidth]{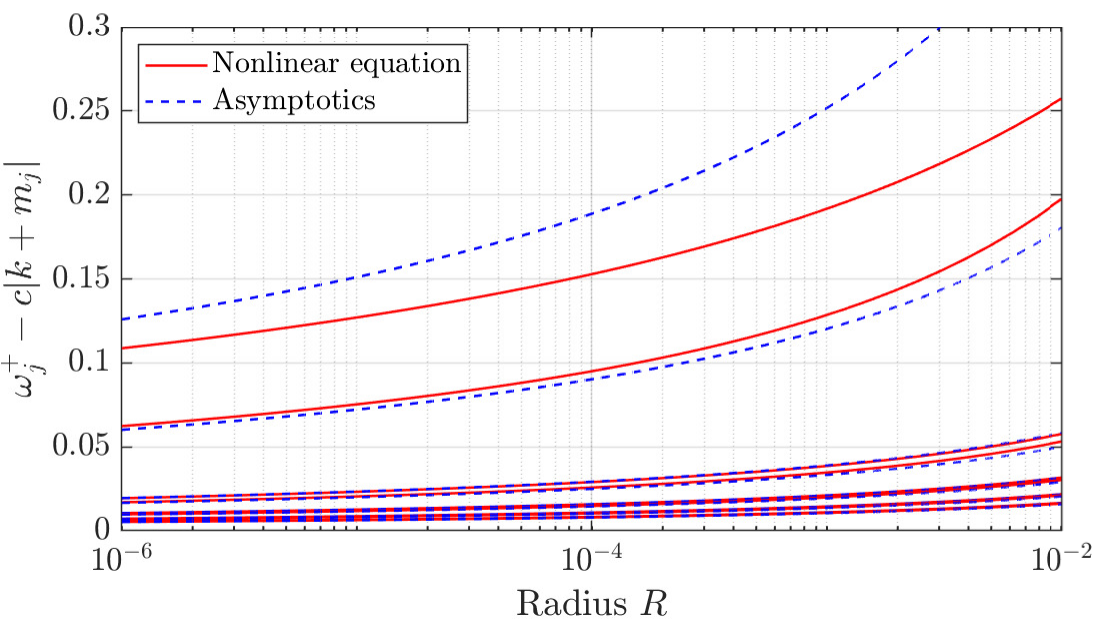}
		\caption{Band functions at fixed $k = 0.9\pi/L$ for varying $R$, in the case $d=1$ and a single inclusion $N=1$, computed through the asymptotic approximation \eqref{eq:approx} and the nonlinear equation \eqref{eq:1dnonlin}. The first band ($q_j = 0$, highest in this figure) shows a large discrepancy between \eqref{eq:approx} and \eqref{eq:1dnonlin}, while the higher bands are closely approximated by \eqref{eq:approx}.}
		\label{fig:nonlin}
\end{figure}

\begin{remark}
	As pointed out in \Cref{rmk:highbands}, the asymptotic formulas provided in \Cref{thm:band}, \Cref{thm:band2d} and \Cref{thm:band1d} fail to capture the bands $\omega_j^+$ above $\Omega$, which in the limit $\epsilon \to 0$ tend to the light cone $\omega = c|\bk+\bq|$ for $\bq\in \Lambda^*$. Nevertheless, these bands can be approximated through the approach of nonlinear equations as outlined above. As we see in \Cref{fig:bandsa}, this approach provides a good approximation even for the higher bands, and shows that there are band gaps which open around the center and edge of the Brillouin zone.
\end{remark}

\begin{remark}
	As demonstrated in \Cref{fig:bandsb}, the frequencies given by the asymptotic formula \eqref{eq:1d} belong to different band functions on either side of the singular point $\omega = c|k|$. For nonzero $\epsilon$, there will be a band gap between these bands. Moreover, for small enough $\epsilon$, the second of these bands falls below $\Omega$. In this case, \Cref{fig:bandsc} shows that there is a solution to the nonlinear equation \eqref{eq:1dnonlin} which passes through $\Omega$. This causes a hybridization between the sequence of bands which accumulate at $\Omega$, resulting in the opening of tiny band gaps. In contrast to \Cref{thm:bandgap_general}, there is no band gap above $\Omega$, rather, there is a band function emerging arbitrarily close to $\Omega$. 
\end{remark}

\subsection{Two dimensions}\label{sec:lattice2d}
In two dimensions, the zeroth-order lattice sum is given by \cite{linton2010lattice} 
\begin{equation}
	S^{\omega/c}_\mathrm{helm}(\bk) = \frac{c^2}{|Y|(\omega^2 - c^2|\bk|^2)} + \frac{1}{|Y|}\sum_{\bq \in \Lambda^*\setminus \{0\}}\frac{c^2|\bq|^2+\omega^2-c^2|\bk+\bq|^2}{|\bq|^2(\omega^2-c^2|\bk+\bq|^2)} + C(\omega),
\end{equation}
where $C(\omega)$ is constant with respect to $\bk$. If $\omega_j$ is close to $c|\bk + \bq_j|$, we have
\begin{equation}\label{eq:Ssing2D}
	S^{\omega/c}_\mathrm{helm}(\bk) =
		\frac{c^2}{|Y|(\omega^2 - c^2|\bk+\bq_j|^2)} +  \widetilde{C}_j(\omega),
\end{equation}
for some $\widetilde{C}_j(\omega)$ which is bounded for $\omega$ in a neighbourhood of $c|\bk + \bq_j|$.

We repeat the approach of \Cref{sec:lattice1d}. In the two-dimensional case, we need to go back to the operator formulation. Take for simplicity $N=1$. Again, we assume 
\begin{equation}\omega_j^\pm(\bk) = c|\bk+\bq_j| + \alpha,\end{equation}
for some $\bq_j \in \Lambda^*$ and some $\alpha \ll 1$. In order to have solutions of order $O(1)$ as $\epsilon \to 0$, we require 
\begin{equation}\label{eq:alpha}
	\alpha = \epsilon \alpha_0,
\end{equation}
for some constant $\alpha_0$. In this case, the singularity of $S_\mathrm{helm}$ causes the operator $\epsilon\A_1^{\bk,\omega}$ to have order $O(\epsilon^{-1})$ as $\epsilon \to 0$. More specifically, 
\begin{equation}
	\A_1^{\bk,\omega}[\phi](\bx) = - \frac{g^2s_0\omega S^{\omega/c}_\mathrm{helm}(\bk)}{|Y|c^2}\int_{B}  \phi(\by)\d \by + O(1).
\end{equation}
Combined with \eqref{eq:expA} and \eqref{eq:Ssing2D}, we find that the full operator $\G^{\bk,\omega}_\epsilon$ satisfies 
\begin{equation}
	S_\epsilon\G^{\bk,\omega}_\epsilon S_\epsilon^{-1} = \widetilde{\A}_0 + O(\epsilon),
\end{equation}
where the limiting operator $\widetilde{\A}_0 =\widetilde{\A}_0(\alpha_0)$ is defined as
\begin{equation}
	 \widetilde{\A}_0[\phi](x) = -(c|\bk+\bq_j|-\Omega)\phi(\bx) -\frac{g^2s_0}{2\pi c}\int_{B} \frac{\phi(\by)}{|\bx-\by|} \d \by 	 - \frac{g^2s_0}{\alpha_0|Y|}\int_{B}  \phi(\by)\d \by.
\end{equation}
We assume $c|\bk+\bq_j|$ is not an eigenvalue of $\A_0^\omega$, so that $(\A_0^{c|\bk+\bq_j|})^{-1}$ exists and is bounded. In this case,  $\widetilde{\A}_0[\phi] = 0$ is equivalent to 
\begin{equation}
	\alpha_0 \phi(\bx) - \frac{g^2s_0}{|Y|}\int_{B}  \left(\A_0^{c|\bk+\bq_j|}\right)^{-1}\phi(\by)\d \by = 0.
\end{equation}
In other words, $\alpha_0$ is the eigenvalue of a rank-1 operator, and is given by
\begin{equation}\label{eq:alpha0}
	\alpha_0 = \frac{g^2s_0}{|Y|}\int_{B}  \left(\A_0^{c|\bk+\bq_j|}\right)^{-1}[\chi_B](\by)\d \by.
\end{equation}
In particular, $\alpha_0$ depends on the shape of $B$. We can easily extract the behaviour of $\alpha_0$ for large frequencies: when $\bq_j$ is large we have
\begin{equation}
	\alpha_0 \approx \frac{g^2s_0}{c|\bq_j|}\frac{|B|}{|Y|}, \quad \text{as } \bq_j \to \infty.
\end{equation}
Moreover, for fixed $\bq_j$, we can extract the behaviour of $\alpha_0$ when $\bk$ approaches a singularity of $\left(\A_0^{c|\bk+\bq_j|}\right)^{-1}$. We assume that $\lambda_\star=c|\bk_\star+\bq_j|$ is a simple eigenvalue of $\A_0^{c|\bk+\bq_j|}$ with corresponding eigenvector $u_\star$. We then have
\begin{equation}
	\left(\A_0^{c|\bk+\bq_j|}\right)^{-1} = \frac{\langle u_\star, \cdot \rangle u_\star}{c|\bk+\bq_j| - \lambda_\star} + \mathcal{R}(\bk),
\end{equation}
where $\mathcal{R}$ is an operator such that $\|\mathcal{R}(\bk)\|$ is bounded for $\bk$ in a neighbourhood of $\bk_\star$. Suppose $\bk-\bk_\star = \delta\bv$ for some fixed $\bv$ with $|\bv| = 1$. Here, $\delta\ll1$ describes the distance between $\bk$ and the singular point $\bk_\star$. We then have
\begin{equation}
	\alpha_0 \approx \begin{cases}\ds \frac{g^2s_0|\bk_\star+\bq_j|}{\delta c|Y|(\bk_\star+\bq_j)\cdot\bv}\left(\int_{B} u_\star \d \by\right)^2,  &(\bk_\star+\bq_j)\cdot \bv\neq 0, \\[0.3em] \ds \frac{g^2s_0|\bk_\star+\bq_j|}{\delta^2c|Y|}\left(\int_{B} u_\star \d \by\right)^2,  &(\bk_\star+\bq_j)\cdot \bv= 0,\end{cases} \qquad \text{as } \delta \to 0.
\end{equation}
In both cases, $\alpha_0$ explodes as $\delta \to 0$, and \eqref{eq:alpha} is not valid if $c|\bk+\bq_j|$ is an eigenvalue of $\A_0^{c|\bk+\bq_j|}$.

\subsection{Three dimensions}\label{sec:lattice3d}
In three dimensions, we have the lattice sum \cite{linton2010lattice}
\begin{equation}
	S^{\omega/c}_\mathrm{helm}(\bk) = \frac{1}{\sqrt{4\pi}|Y|}\sum_{\bq \in \Lambda^*}\frac{c^2}{\omega^2-c^2|\bk+\bq|^2} + C(\omega),
\end{equation}
for some $C$ which is constant in $\bk$ and bounded for $\omega$ in compact subsets of $\R\setminus\{0\}$. If $\omega_j$ is close to $c|\bk + \bq_j|$, we have
\begin{equation}
	S^{\omega/c}_\mathrm{helm}(\bk) =
	\frac{c^2}{\sqrt{4\pi}|Y|(\omega^2 - c^2|\bk+\bq_j|^2)} +  \widetilde{C}_j(\omega),
\end{equation}
for some $\widetilde{C}_j(\omega)$ which is bounded for $\omega$ in a neighbourhood of $c|\bk + \bq_j|$. Again, for simplicity, we take $N=1$, which gives us
\begin{equation}
	\widetilde{\A}_{0}[\phi](\bx) = \A^{c|\bk+\bq_j|}_0  - \frac{g^2s_0}{\sqrt{4\pi}\alpha_0|Y|}\int_{B}  \phi(\by)\d \by.
\end{equation}
As in the two-dimensional case, we get
\begin{equation}
	\alpha_0 = \frac{g^2s_0}{\sqrt{4\pi}|Y|}\int_{B}  \left(\A_0^{c|\bk+\bq_j|}\right)^{-1}[\chi_B](\by)\d \by.
\end{equation}

\section{Conclusions}

In this paper we have continued our study of the quantum optics of a single photon interacting with a system of two level atoms begun in ~part I. In this work we considered the case that the density of atoms $\rho(\bx)$ respects the symmetries of a general $n$ dimensional lattice. We proved a general structure theorem on the spectrum of the corresponding Hamiltonian which shows it consists of two families of eigenvalues: one which accumulates at the resonant frequency $\Omega$ and one which increases to $+\infty$. Moreover, we have included a detailed study of the case that $\rho(\bx)$ is piecewise constant in what we call the \textit{high-contrast} regime: The nonzero values of $\rho$ are large while the volume of the support is small. In this class of examples, we obtain a complete asymptotic expansion of the band functions as perturbations of flat bands. The case of atom-field interaction with multiple photons was introduced in \cite{kraisler2022kinetic} and is a highly interesting direction for future work. Several other interesting questions we intend to pursue in further work include:

\begin{enumerate}
    \item \Cref{thm:bandgap_general} guarantees the existence of band gaps above $\Omega$ for positive $\rho$. What can be said, in general, about the existence of other gaps in the band structure?
    \item In \Cref{sec:nonlin}, formal arguments are used to study the bands near the Rayleigh singularities. Can one turn these into rigorous arguments?
    \item In this paper we studied spectral phenomena, but of clear interest are dynamical questions. What effective dynamics, if any, does a wave packet centered around the accumulation at $\Omega$ satisfy asymptotically?
\end{enumerate}
\appendix

\section{Free-space Green's function} \label{sec:G}
The free-space Greens's function $G^k(\bx)$  satisfies
\begin{equation}\left((-\Delta)^{1/2} - k \right)G^k(\bx) =  \delta(\bx),\end{equation}
and can be represented as
\begin{equation}\label{eq:Gint}
	G^k(\bx) = \frac{1}{(2\pi)^d}\int_{\R^d} \frac{e^{\iu \bk\cdot \bx}}{|\bk|-k}\d \bk.
\end{equation}
Based on the elementary identity 
\begin{equation} \frac{1}{|\bk|-k} = \frac{1}{|\bk| + k} + \frac{2k}{|\bk|^2 - k^2}.\end{equation}
we find that, for $k>0$, 
\begin{equation} \label{eq:Gid}
	G^k(\bx) = G^{-k}(\bx) + 2kG^k_\mathrm{helm}(\bx).
\end{equation}
where $G^k_\mathrm{helm}$ denotes the Helmholtz Green's function. We have the following integral representations of $G^{-k}$ part I:
\begin{align}
	G^{-k}(x) &= \frac{1}{\pi}\int_{0}^\infty \frac{r e^{- r}}{r^2 + k^2x^2}\d r, &   &d=1, \label{eq:Gneg1d}\\ 
	G^{-k}(\bx) &= \frac{1}{2\pi|\bx|} \int_0^\infty \frac{r^2e^{-r}}{r^2 + k^2|\bx|^2\left(1+\sqrt{1+\frac{r^2}{k^2|\bx|^2}}\right)} \d r, &  &d=2, \label{eq:Gneg2d} \\ 
	G^{-k}(\bx) &= \frac{1}{2\pi^2|\bx|^2}\int_{0}^\infty \frac{r^2e^{- r}}{r^2 + k^2|\bx|^2}\d r, &  &d=3. \label{eq:Gneg3d} 
\end{align}
In all three cases, we have have that $G^{-k}(\bx)> 0$ for $\bx\in \R^d\setminus\{0\}$. Moreover, we have 
\begin{equation} \label{eq:decay}
G^{-k}(\bx) = O\left(|\bx|^{-(d+1)}\right), \qquad |\bx| \to \infty.\end{equation}

\section{Pseudo-periodic Green's function}
The pseudo-periodic Green's function $G^{\bk,k}(\bx)$ satisfies 
\begin{equation}
\left((-\Delta)^{1/2} - k \right)G^k(\bx) =   \sum_{\bm \in \Lambda}e^{\iu \bk\cdot \bm} \delta(\bx-\bm),
\end{equation}
and can be represented as
\begin{equation}\label{eq:spatial_app}
	G^{\bk,k}(\bx) = \sum_{\bm \in \Lambda}e^{\iu \bk\cdot \bm} G^{k}(\bx-\bm),
\end{equation}
or equivalently as 
\begin{equation}\label{eq:spectral_app}
G^{\bk,\omega/c}(\bx) = \sum_{\bq\in \Lambda^*} \frac{e^{\iu(\bq+\bk)\cdot \bx}}{|\bk+\bq| - \omega/c}.\end{equation}
In light of \eqref{eq:Gid}, for $k>0$ we have 
\begin{equation}
G^{\bk,k}(\bx) = G^{\bk,-k}(\bx) + 2kG^{\bk,k}_\mathrm{helm}(\bx),
\end{equation}
where $G^{\bk,k}_\mathrm{helm}(\bx)$ is the pseudo-periodic Green's function for the Helmholtz equation given by
\begin{equation}
G^{\bk,k}_\mathrm{helm}(\bx) = \sum_{\bm \in \Lambda}e^{\iu \bk\cdot \bm} G^{\omega/c}_{\mathrm{helm}}(\bx-\bm),
\end{equation}
while $G^{\bk,-k}(\bx)$ is given by the following series
\begin{equation}
G^{\bk,-k}(\bx) = \sum_{\bm \in \Lambda}e^{\iu \bk\cdot \bm} G^{-k}(\bx-\bm).
\end{equation}
The convergence theory and numerical evaluation of $G^{\bk,k}_\mathrm{helm}$ is well-studied (see, for example, \cite{linton2010lattice} for a comprehensive review); if $k \neq |\bk+\bq|$ for all $\bq \in \Lambda^*$, this series converges uniformly for $\bx$ in compact sets of $\R^d$, $\bx\neq 0$. Moreover, from \eqref{eq:decay} we know that the series for $G^{\bk,-k}(\bx)$ is absolutely and uniformly convergent for $\bx\neq 0$, and can be numerically evaluated by truncating the series.

\subsection{Singularity of the Green's function} \label{sec:sing}
In this section, we report the behaviour of $G^{\bk,k}(\bx)$ in the case $\bx \to 0$. Throughout, we let $B\subset \R^d$ be a fixed, bounded domain and $\epsilon \ll 1$. We additionally assume $k \neq |\bk+\bq|$ for all $\bq \in \Lambda^*$.

From \eqref{eq:spatial_app} we know that 
\begin{equation}\label{eq:G0+Gk}
G^{\bk,k}(\epsilon\bx) = G^{k}(\epsilon\bx) + 
\sum_{\bm \in \Lambda\setminus \{0\} }e^{\iu \bk\cdot \bm} G^{k}(\epsilon\bx-\bm),
\end{equation}
where the latter term is bounded as $\epsilon \to 0$. In other words, $G^{\bk,k}(\epsilon\bx)$ inherits the singularity of $G^{k}(\epsilon\bx)$, and from part I we have 
\begin{equation}
	\epsilon^dG^{\bk,k}(\epsilon \bx) = \sum_{n=0}^\infty \epsilon^{n+1}A_n^{\bk,k}(\bx) +  \sum_{n=d-1}^\infty \epsilon^{n+1}\log(\epsilon)B_n^k(\bx),
\end{equation}
for functions $A_n^{\bk,k}(\bx)$ and $B_n^k(\bx)$ which can be explicitly computed. In particular, from \eqref{eq:G0+Gk} we know that $A_n^{\bk,k}$ is independent of $\bk$ for $n<d-1$, and $A_{d-1}^{\bk,k}$ has a $\bk$-dependent term given by
\begin{equation}\label{eq:Sapp}
S^{k}(\bk) = \sum_{\bm \in \Lambda\setminus\{0\} }e^{\iu \bk\cdot \bm} G^{k}(\bm),
\end{equation}
known as the \emph{zeroth order lattice sum}. With \eqref{eq:Sapp}, along with the expansions of $G^k$ from  part I, we can find the the first few functions $A_n^k$ and $B_n^k$, which are used in the current work. In one spatial dimension, we have 
	\begin{align}
		A_0^{\bk,k}(x) = -\frac{1}{\pi} \left(\log(k|x|)+ \gamma\right) + \iu + S^{k}(\bk), \quad B_{0}^k(x) =-\frac{1}{\pi}.
	\end{align}
Similarly, in two spatial dimensions we have 
\begin{equation}A_0(\bx) = \frac{1}{2\pi|\bx|}, \ B_{1}^k = -\frac{k}{2\pi}, \ A_1^{\bk,k}(\bx) = -\frac{k}{2\pi} \left(\log(k|\bx|) + \gamma \right) + \frac{\iu k}{2} + S^{k}(\bk), \ B_{2} = 0, \ A_2^k = -\frac{k}{\pi}.\end{equation}
Finally, in three spatial dimensions, we have 
\begin{equation}\label{eq:G012}
	A_0(\bx) = \frac{1}{2\pi^2|\bx|^2}, \ A_1^k(\bx) = \frac{k}{2\pi|\bx|}, \
	A_2^{\bk,k}(\bx) = \frac{k^2}{2\pi^2}\bigl(1-\gamma + \pi\iu - \log\left(k|\bx|\right) \bigr) + S^{k}(\bk), \ B_{2}^k = -\frac{k^2}{2\pi^2}. 
\end{equation}
For shifted origins of $G^{\bk,\omega/c}$, we know that $G^{\bk,\omega/c}(\epsilon \bx + \bz_i - \bz_j)$ is analytic as function of $\bx$, so we have Taylor expansions
\begin{equation}
	\epsilon^dG^{\bk,\omega/c}(\epsilon \bx + \bz_i - \bz_j) = \sum_{n=d-1}^\infty \epsilon^{n+1}A_{n,(i,j)}^{\bk,\omega/c}(\bx), \qquad i,j=1,...,N, \quad i\neq j,
\end{equation}
for functions $A_{n,(i,j)}^{\bk,\omega/c}(\bx)$; in particular, $A_{n,(i,j)}^{\bk,\omega/c}(\bx) = 0$ for $n <d-1$ and 
\begin{equation}
A_{d-1,(i,j)}^{\bk,\omega/c}(\bx) =  G^{\bk,\omega/c}(\bz_i - \bz_j).
\end{equation}

\subsection{Pseudo-periodic Green's function in one dimension}
In one spatial dimension we have an explicit expression of the pseudo-periodic Helmholtz Green's function $G^{k,\omega/c}_{\mathrm{helm}}(x)$. The fractional Green's function is now given by part I 
\begin{equation}G^{\omega/c}(x) = \frac{e^{\iu {\omega} |x|/c }}{2\pi}E_1\left(\frac{\iu \omega|x|}{c}\right) + \frac{e^{-\iu \omega |x| / c}}{2\pi}E_1\left(-\frac{\iu \omega|x|}{c}\right) + \iu e^{\iu \omega|x|/c},\end{equation}
and the Helmholtz Green's function
\begin{equation}G_{\mathrm{helm}}^{\omega/c}(x) = \frac{\iu c}{2\omega}e^{\iu \omega |x|/c}.\end{equation}
\begin{lemma}\label{lem:Sd1}
	Let $L$ be the cell length in a one-dimensional lattice in one dimension. Then the pseudo-periodic Green's function $G^{k,\omega/c}_{\mathrm{helm}}(x)$ is given by
	\begin{equation}G^{k,\omega/c}_{\mathrm{helm}}(x) = \frac{\iu c}{2\omega}\left( e^{\iu \omega|x|/c}+  \frac{e^{\iu (\omega/c+k)L}}{1-e^{\iu (\omega/c+k)L}}e^{\iu \omega x/c} + \frac{e^{\iu (\omega/c-k)L}}{1-e^{\iu (\omega/c-k)L}}e^{-\iu \omega x/c}\right).\end{equation}
\end{lemma}
\begin{proof}
	Defining $u(x) = G^{k,\omega/c}_{\mathrm{helm}}(x) - G^{\omega/c}_{\mathrm{helm}}(x)$, we know that $u'' + \frac{\omega^2}{c^2}u =0$. Therefore we can write
	\begin{equation}G^{k,\omega/c}_{\mathrm{helm}}(x) = G^{\omega/c}_{\mathrm{helm}}(x) +  S_1(k) e^{\iu \omega x/c} + S_2(k)e^{-\iu \omega x/c},\end{equation}
	where $S_1$ and $S_2$ are constant in $x$. We have that 
	\begin{equation}S_1 + S_2 = u(0), \quad S_1 - S_2 = \frac{u'(0)c}{\iu \omega}.\end{equation}
	We have
	\begin{equation}u(0) =  \frac{\iu c}{2\omega}\sum_{m\in \Lambda, m > 0} e^{\iu (\omega/c+k) m} + e^{\iu (\omega/c-k) m}, \quad \frac{u'(0) c}{\iu \omega} =   \frac{\iu c}{2\omega}\sum_{m\in \Lambda, m > 0} e^{\iu (\omega/c+k) m} - e^{\iu (\omega/c-k) m}\end{equation}
	These sums are given by the geometric series:
	\begin{equation}u(0) =  \frac{\iu c}{2\omega}\sum_{m\in \Lambda, m > 0} e^{\iu (\omega/c+k) m} + e^{\iu (\omega/c-k) m} = \frac{\iu c}{2\omega}\left(\frac{e^{\iu (\omega/c+k)L}}{1-e^{\iu (\omega/c+k)L}} + \frac{e^{\iu (\omega/c-k)L}}{1-e^{\iu (\omega/c-k)L}}\right),\end{equation}
	and
	\begin{equation}\frac{u'(0) c}{\iu \omega} =   \frac{\iu c}{2\omega}\sum_{m\in \Lambda, m > 0} e^{\iu (\omega/c+k) m} - e^{\iu (\omega/c-k) m} =  \frac{\iu c}{2\omega}\left(\frac{e^{\iu (\omega/c+k)L}}{1-e^{\iu (\omega/c+k)L}} - \frac{e^{\iu (\omega/c-k)L}}{1-e^{\iu (\omega/c-k)L}}\right) \end{equation}
	Therefore
	\begin{equation}S_1 = \frac{\iu c}{2\omega}\frac{e^{\iu (\omega/c+k)L}}{1-e^{\iu (\omega/c+k)L}}, \quad S_2 =  \frac{\iu c}{2\omega}\frac{e^{\iu (\omega/c-k)L}}{1-e^{\iu (\omega/c-k)L}},\end{equation}
	which proves the claim.
\end{proof}

Again have a decomposition of the lattice sum
\begin{align} 
	S^{\omega/c}(k) &= \sum_{m \in \Lambda\setminus\{0\} }e^{\iu km} G^{\omega/c}(m) \\
	&= \frac{2\omega}{c} S_{\mathrm{helm}}^{\omega/c}(k) + S^{-\omega/c}(k), \label{eq:Sdecomp}
\end{align}
where $S_{\mathrm{helm}}^{\omega/c}(k)$ is the lattice sum associated with the Helmholtz Green's function and $S^{-\omega/c}(k)$ corresponds to the absolutely convergent sum of the remaining $O(m^{-2})$-term. From \Cref{lem:Sd1} we have
\begin{equation}\label{eq:Sh}
	 S_{\mathrm{helm}}^{\omega/c}(k) = \frac{\iu c}{2\omega}\left(\frac{e^{\iu (\omega/c+p)L}}{1-e^{\iu (\omega/c+p)L}} + \frac{e^{\iu (\omega/c-k)L}}{1-e^{\iu (\omega/c-k)L}} \right) = \frac{\iu c}{2\omega}\frac{\cos(kL)-e^{\iu \omega L/c}}{\cos\left(\frac{\omega L}{c}\right) - \cos(kL)}.\end{equation}

\subsection{Lower bound of the band structure} \label{sec:propproof}
We now present the proof of \Cref{prop:lower}, which relies on the spectral representation \eqref{eq:spectral_app} of the Green's function $G^{\bk,\omega/c}$.

\begin{proof}[Proof of \Cref{prop:lower}]
For $\omega <0$, we can invert the left-hand side of \eqref{eq:EVproblem} and we have a Lippmann-Schwinger representation of the form
	\begin{equation}\psi(\bx) = -\frac{g^2}{c(\omega-\Omega)}\int_{Y} G^{\bk,\omega/c}(\bx-\by) \rho(\by) \psi(\by) \d \by.\end{equation}
	If we integrate over $Y$ and estimate the integral, we therefore have
	\begin{equation}\label{eq:est}\int_{Y} |\psi| \d \bx \leq -\frac{g^2}{c(\omega-\Omega)}\left(\int_{Y}\left|G^{\bk,\omega/c}(\bx)\right|\d \bx\right)\left(\int_{Y}\rho|\psi| \d \bx\right).\end{equation}
 From \eqref{eq:spectral} we have that
    \begin{equation}\left|G^{\bk,\omega/c}(\bx)\right| = \left|\sum_{\bq\in \Lambda^*} \frac{e^{\iu(\bq+\bk)\cdot \bx}}{|\bk+\bq| - \omega/c}\right| = \sum_{\bq\in \Lambda^*} \frac{e^{\iu\bq\cdot \bx}}{|\bk+\bq| - \omega/c},
    \end{equation}
    which follows from the fact that the latter series is real and positive for $\omega <0$. We therefore have that
    \begin{equation}
    \int_{Y}\left|G^{\bk,\omega/c}(\bx)\right|\d \bx = \frac{1}{|\bk| - \omega/c} \leq -\frac{c}{\omega}.
    \end{equation}
    We then have from \eqref{eq:est} that 
	\begin{equation}\int_{Y} |\psi| \d \bx \leq \frac{g^2\|\rho\|_{\infty}}{\omega(\omega-\Omega)}\int_{Y}|\psi| \d \bx,\end{equation}
 and since $\psi\ne0$, we find that $g^2\|\rho\|_{\infty} \geq \omega(\omega-\Omega)$. From this inequality, we find that 
\begin{equation}
\omega_1^-(\bk) \geq \frac{\Omega - \sqrt{\Omega^2 + 4g^2\|\rho\|_{\infty} }}{2},
\end{equation} 
which finishes the proof. 
\end{proof}
\bibliographystyle{abbrv}
\bibliography{one-photon}{}

\begin{thebibliography}{10}

\bibitem{ammari2018mathematical}
H.~Ammari, B.~Fitzpatrick, H.~Kang, M.~Ruiz, S.~Yu, and H.~Zhang.
\newblock {\em Mathematical and Computational Methods in Photonics and
  Phononics}, volume 235 of {\em Mathematical Surveys and Monographs}.
\newblock American Mathematical Society, Providence, 2018.

\bibitem{ammari2004splitting}
H.~Ammari and F.~Triki.
\newblock Splitting of resonant and scattering frequencies under shape
  deformation.
\newblock {\em Journal of Differential equations}, 202(2):231--255, 2004.

\bibitem{barik2018topological}
S.~Barik, A.~Karasahin, C.~Flower, T.~Cai, H.~Miyake, W.~DeGottardi, M.~Hafezi,
  and E.~Waks.
\newblock A topological quantum optics interface.
\newblock {\em Science}, 359(6376):666--668, 2018.

\bibitem{bloch2012quantum}
I.~Bloch, J.~Dalibard, and S.~Nascimbene.
\newblock Quantum simulations with ultracold quantum gases.
\newblock {\em Nature Physics}, 8(4):267--276, 2012.

\bibitem{de2008solid}
H.~De~Riedmatten, M.~Afzelius, M.~U. Staudt, C.~Simon, and N.~Gisin.
\newblock A solid-state light--matter interface at the single-photon level.
\newblock {\em Nature}, 456(7223):773--777, 2008.

\bibitem{Gohberg1971}
I.~Gohberg and E.~Sigal.
\newblock An operator generalization of the logarithmic residue theorem and the
  theorem of {R}ouch\'{e}.
\newblock {\em Sbornik Mathematics}, 13(4):603--625, 1971.

\bibitem{one-photon_bound}
E.~O. Hiltunen, J.~Kraisler, J.~C. Schotland, and M.~I. Weinstein.
\newblock Nonlocal pdes and quantum optics: Bound states and resonances.
\newblock {\em arXiv preprint arXiv:2306.10431}, 2023.

\bibitem{hoskins_2021}
J.~Hoskins, J.~Kaye, M.~Rachh, and J.~Schotland.
\newblock Analysis of single-excitation states in quantum optics.
\newblock {\em arXiv preprint arXiv:2110.07049}, 2021.

\bibitem{hoskins_2023}
J.~Hoskins, J.~Kaye, M.~Rachh, and J.~C. Schotland.
\newblock A fast, high-order numerical method for the simulation of
  single-excitation states in quantum optics.
\newblock {\em Journal Computational Physics}, 473:111723, 2023.

\bibitem{john1991quantum}
S.~John and J.~Wang.
\newblock Quantum optics of localized light in a photonic band gap.
\newblock {\em Physical Review B}, 43(16):12772, 1991.

\bibitem{kraisler_2022}
J.~Kraisler and J.~C. Schotland.
\newblock Collective spontaneous emission and kinetic equations for one-photon
  light in random media.
\newblock {\em Journal of Mathematical Physics}, 63:031901, 2022.

\bibitem{kraisler2022kinetic}
J.~Kraisler and J.~C. Schotland.
\newblock Kinetic equations for two-photon light in random media.
\newblock {\em Journal of Mathematical Physics}, 64:111903, 2023.

\bibitem{linton2010lattice}
C.~M. Linton.
\newblock Lattice sums for the helmholtz equation.
\newblock {\em SIAM review}, 52(4):630--674, 2010.

\bibitem{perczel2017topological}
J.~Perczel, J.~Borregaard, D.~E. Chang, H.~Pichler, S.~F. Yelin, P.~Zoller, and
  M.~D. Lukin.
\newblock Topological quantum optics in two-dimensional atomic arrays.
\newblock {\em Physical review letters}, 119(2):023603, 2017.

\bibitem{petit2013electromagnetic}
R.~Petit~(editor).
\newblock {\em Electromagnetic Theory of Gratings}, volume~22 of {\em Topics in
  Current Physics}.
\newblock Springer Science \& Business Media, 2013.

\bibitem{yablonovitch1993photonic}
E.~Yablonovitch.
\newblock Photonic band-gap structures.
\newblock {\em JOSA B}, 10(2):283--295, 1993.

\end{thebibliography}

\end{document}